\documentclass{article}

\usepackage[utf8]{inputenc}
\usepackage{amsmath,amsthm,amssymb}
\usepackage{graphicx}
\usepackage{hyperref}
\usepackage{paralist}
\usepackage{todonotes}
\usepackage{xspace}
\usepackage{booktabs}
\usepackage{multirow}
\usepackage{caption}
\usepackage{subcaption}

\theoremstyle{plain}
\newtheorem{theorem}{Theorem}
\newtheorem{lemma}{Lemma}

\newtheorem{corollary}{Corollary}

\DeclareMathOperator{\rot}{rot}

\newcommand{\cupdot}{\mathaccent\cdot\cup}

\newcommand{\eps}{\varepsilon}

\newcommand{\inedge}{{\mathrm{in}}}
\newcommand{\outedge}{{\mathrm{out}}}

\newcommand{\true}{\texttt{true}\xspace}

\makeatletter
\DeclareRobustCommand{\bfseries}{%
   \not@math@alphabet\bfseries\mathbf
   \fontseries\bfdefault\selectfont
   \boldmath
}
\makeatother

\title{Complexity of Higher-Degree Orthogonal Graph Embedding in the
  Kandinsky Model}

\author{Thomas Bläsius\thanks{Faculty of Informatics, Karlsruhe
    Institute of Technology, Karlsruhe, \texttt{blaesius@kit.edu,
      guido.brueckner@student.kit.edu}} \and Guido Brückner$^*$ \and
  Ignaz Rutter\thanks{Faculty of Informatics, Karlsruhe Institute of
    Technology, Karlsruhe and Department of Applied Mathematics,
    Faculty of Mathematics and Physics, Charles University, Prague,
    \texttt{rutter@kit.edu}}}

 \date{}

\begin{document}

\maketitle

\begin{abstract}
  We show that finding orthogonal grid-embeddings of plane graphs
  (planar with fixed combinatorial embedding) with the minimum number
  of bends in the so-called Kandinsky model (which allows vertices of
  degree $> 4$) is NP-complete, thus solving a long-standing open
  problem.  On the positive side, we give an efficient algorithm for
  several restricted variants, such as graphs of bounded branch width
  and a subexponential exact algorithm for general plane graphs.
\end{abstract}

\section{Introduction}
\label{sec:introduction}

Orthogonal grid embeddings are a fundamental topic in computer science
and the problem of finding suitable grid embeddings of planar graphs
is a subproblem in many applications, such as graph
visualization~\cite{t-hgdv-13} and VLSI
design~\cite{l-aeglv-80,v-ucvc-81}.  Aside from the area requirement,
the typical optimization goal is to minimize the number of bends on
the edges (which also heuristically minimizes the area).
Traditionally, grid embeddings have been studied for graphs with
maximum degree~4, which is natural since it allows to represent
vertices by grid points and edges by internally disjoint chains of
horizontal and vertical segments on the grid.  For a fixed
combinatorial embedding Tamassia showed that the number of bends in
the grid embedding can be efficiently minimized~\cite{gt-ccurpt-01};
the running time was recently reduced to $O(n^{1.5})$~\cite{ck-am-12}.
In contrast, if the combinatorial embedding is not fixed, it is
NP-complete to decide whether a 0-embedding (a \emph{$k$-embedding} is
a planar orthogonal grid embedding with at most $k$ bends per edge)
exists~\cite{gt-ccurpt-01}, thus also showing that bend minimization
is NP-complete and hard to approximate within a factor of
$O(n^{1-\eps})$.  In contrast, a 2-embedding exists for every graph
except the octahedron~\cite{bk-bhod-98}.  Recently it was shown that
the existence of a 1-embedding can be tested
efficiently~\cite{bkrw-odfc-14}.  The problem is FPT if some subset of
size $k$ has to have $0$ bends~\cite{blr-ogdie-14}.  If there are no
$0$-bend edges, it is even possible to minimize the number of bends in
the embedding, where the first bend on each edge is not
counted~\cite{brw-oodc-13}.

The main drawback of all these results is that they only apply to
graphs of maximum degree~4.  There have been several suggestions for
possible generalizations to allow vertices of higher
degree~\cite{km-qodpg-98,tbb-agdrd-88}.  For example, it is possible
to model higher-degree vertices by boxes whose shape is restricted to
a rectanlge.  The disadvantage is that, in this way, the vertices may
be stretched arbitrarily in order to avoid bends.  In particular, a
visibility representation of a graph can be interpreted as a
$0$-embedding in this model (and such a representation exists for
every planar graph).  It is thus natural to forbid stretching of
vertices.

Fößmeier and Kaufmann~\cite{fk-ddgn-95} proposed a generalization of
planar orthogonal grid embeddings, the so-called Kandinsky model
(originally called podevsnef), that overcomes this problem and
guarantees that vertices are represented by boxes of uniform size.
Essentially their model allows to map vertices to grid points on a
coarse grid, while routing the edges on a much finer grid.  The
vertices are then interpreted as boxes on the finer grid, thus
allowing several edges to emanate from the same side of a vertex; see
Section~\ref{sec:preliminaries} for a precise definition.  Fößmeier
and Kaufmann model the bend minimization in the fixed combinatorial
embedding setting by a flow network similar to the work of
Tamassia~\cite{t-emn-87} but with additional constraints that limit
the total amount of flow on some pairs of edges.  Fößmeier et
al.~\cite{fkk-2dpg-97} later showed that every planar graph admits a
1-embedding in this model.  Concerning bend minimization, reductions
of the mentioned flow networks to ordinary minimum cost flows have
been claimed both for general bend minimization~\cite{fk-ddgn-95} and
for bend minimization when every edge may have at most one
bend~\cite{fkk-2dpg-97}. 

However, Eiglsperger~\cite{e-aldta-03} pointed out that the reductions
to minimum cost flow is flawed and gave an efficient
$2$-approximation.  Bertolazzi et al.~\cite{bdd-codmnb-00} introduced
a restricted variant of the Kandinsky model (which in general requires
more bends), for which bend minimization can be done in polynomial
time.  Although the Kandinsky model has been later vastly generalized,
e.g., to apply to the layout of UML class
diagrams~\cite{egkkj-aldo-04}, the fundamental question about the
complexity of the bend minimization problem in the Kandinsky model has
remained open for almost two decades.

\paragraph{Contribution and Outline.} 

In this work, we show that the bend minimization problem in the
Kandinsky model is NP-complete even for graphs with a fixed
combinatorial embedding (no matter if we allow or forbid so called
empty faces).  This also holds if each edge may have at most one bend;
see Section~\ref{sec:hardness}.  As an intermediate step, we show
NP-hardness of the problem \textsc{Orthogonal 01-Embeddability}, which
asks whether a plane graph (with maximum degree~4) admits a grid
embedding when requiring some edges to have exactly one and the
remaining edges to have zero bends.  This is an interesting result on
its own, as it can serve as tool to show hardness of other grid
embedding problems.  In particular, it gives a simpler proof for the
hardness of deciding~$0$-embeddability in classic grid embeddings for
graphs with a variable combinatorial embedding.

We then study the complexity of the problem subject to structural
graph parameters in Section~\ref{sec:subexp-algor}.  For graphs with
branch width~$k$, we obtain an algorithm with running time~$2^{O(k
  \log n)}$.  For fixed branch width this yields a polynomial-time
algorithm (running time $O(n^3)$ for series-parallel graphs), for
general plane graphs the result is an exact algorithm with
subexponential running time $2^{O(\sqrt{n}\log n)}$.

\section{Preliminaries}
\label{sec:preliminaries}

The graphs we consider are always \emph{plane}, i.e., they are planar
and have a fixed combinatorial embedding.  A planar graph is 4-planar
if it has maximum degree~4.  It is 4-plane, if it has a fixed
combinatorial embedding.

\subsection{Kandinsky Embedding}
\label{sec:kandinsky-embedding}

Let $G$ be a plane graph.  An \emph{orthogonal embedding} of $G$ maps
each vertex to a grid point and each edge to a path in the grid such
that the resulting drawing is planar and respects the combinatorial
embedding of $G$; see Figure~\ref{fig:kandinsky-embedding}a for an
example.  Clearly, $G$ admits an orthogonal embedding if and only if
no vertex has degree larger than~$4$.  The Kandinsky model introduced
by Fößmeier and Kaufmann~\cite{fk-ddgn-95} is a way to overcome this
limitation.  A \emph{Kandinsky embedding} of $G$ maps each vertex to a
box of constant size centered at a grid point and each edge to a path
in a finer grid such that the resulting drawing is planar and respects
the combinatorial embedding of $G$; see
Figure~\ref{fig:kandinsky-embedding}b for an example.  In a
Kandinsky embedding, a face is \emph{empty} if it does not include a
grid cell of the coarser grid; see
Figure~\ref{fig:kandinsky-embedding}c.  Empty faces are empty in the
sense that there is not enough space to add a vertex inside.  Usually,
one forbids empty faces in Kandinsky embeddings as allowing empty
faces requires a special treatment for faces of size~$3$ compared to
larger faces and cycles that are no faces.  In the following, we
always assume that empty faces are forbidden except when explicitly
allowing them.

\begin{figure}
  \centering
  \includegraphics[page=1]{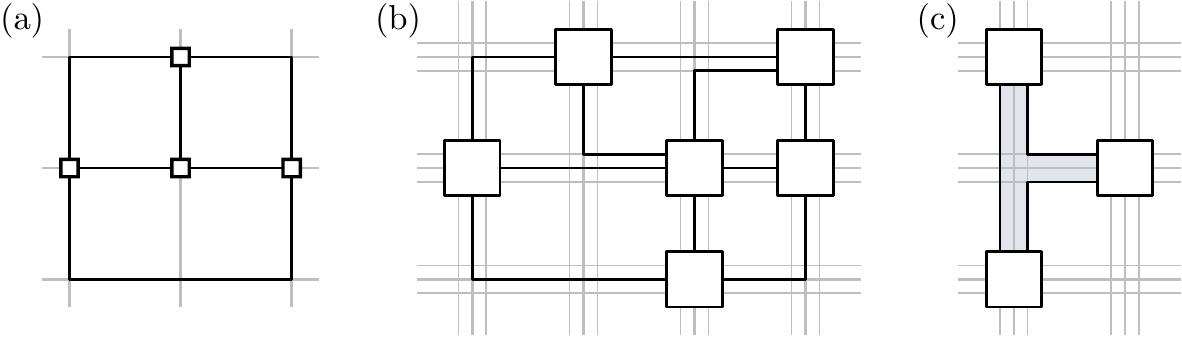}
  \caption{(a)~An orthogonal embedding of the $K_4$.  (b)~A Kandinsky
    embedding of the wheel of size~$5$.  (c)~A Kandinsky embedding
    with an empty face.}
  \label{fig:kandinsky-embedding}
\end{figure}

Every Kandinsky embedding has the so called \emph{bend-or-end
  property}, which can be stated as follows.  One can declare a bend
on an edge $e = uv$ to be \emph{close} to $v$ if it is the first bend
when traversing $e$ from $v$ to $u$ with the additional requirement
that a bend cannot be close to both endpoints $u$ and $v$.  The
bend-or-end property requires that an angle of $0^\circ$ between edges
$uv$ and $vw$ in the face $f$ implies that at least one of the edges
$uv$ and $vw$ has a bend close to $v$ that is concave in $f$
($270^\circ$ angle).  Note that the triangle in
Figure~\ref{fig:kandinsky-embedding}c does not have this property as
the two concave bends cannot be close to all three vertices with
$0^\circ$ angles.

\subsection{Kandinsky Representation}
\label{sec:kand-repr}

A Kandinsky embedding of a planar graph $G$ can be specified in three
stages.  First, its topology is fixed by choosing a combinatorial
embedding of $G$ (which we assume to get with the input).  Second, its
shape in terms of angles between edges and sequences of bends on edges
is fixed.  Third, the geometry is fixed by specifying integer
coordinates for all vertices and bend points.  In analogy to the
definition of combinatorial embeddings as equivalence classes of
planar drawings with the same topology, one can define \emph{Kandinsky
  representations} as equivalence classes of Kandinsky embeddings with
the same topology and the same shape.  As the number of bends (and
thus the cost of an embedding) depends only on the shape and not on
the geometry, we can focus on finding Kandinsky representations and
thus neglect the geometry (at least if we make sure that every
Kandinsky representation has a geometric realization as a Kandinsky
embedding).  For orthogonal embeddings, this approach was introduced
by Tamassia~\cite{t-emn-87}.  It was extended to Kandinsky embeddings
by Fößmeier and Kaufmann~\cite{fk-ddgn-95}.

Let $G$ be a planar graph with the Kandinsky embedding $\Gamma$.  Let
$f$ be a face with the edge $e_1$ in its boundary and let $e_2$ be the
successor of $e_1$ in clockwise direction (counter-clockwise if $f$ is
the outer face).  Let further $v$ be the vertex between $e_1$ and
$e_2$ and let $\alpha$ be the angle at $v$ in $f$.  We define the
\emph{rotation} $\rot_f(e_1, e_2)$ between $e_1$ and $e_2$ to be
$\rot_f(e_1, e_2) = 2-\alpha/90^\circ$; see
Figure~\ref{fig:ortho-rep-1}a.  The rotation $\rot_f(e_1,e_2)$ can
be interpreted as the number of right turns between the edges $e_1$
and $e_2$ at the vertex $v$ in the face $f$.  Note that $e_1 = e_2$ if
$v$ has degree~1, which yields $\rot_f(e_1,e_2) = -2$.  In case it is
clear from the context which two edges are meant when referring to the
vertex $v$ in the face $f$, we also write $\rot_f(v)$ instead of
$\rot_f(e_1, e_2)$ and call it the \emph{rotation of $v$ in $f$}.

\begin{figure}
  \centering
  \includegraphics[page=1]{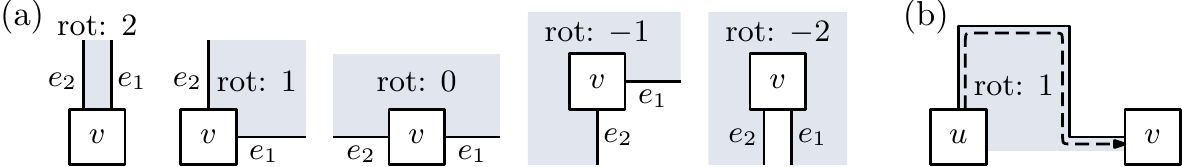}
  \caption{(a)~The possible rotations at a vertex in the face $f$
    (shaded blue). (b)~The rotation of an edge.}
  \label{fig:ortho-rep-1}
\end{figure}

The shape of every edge can also be described in terms of its
rotation.  Let $u$ be a vertex in the boundary of the face $f$ and let
$v$ be its successor in clockwise direction (counter-clockwise if $f$
is the outer face).  Let further $e = uv$ be the corresponding edge.
The \emph{rotation} $\rot_f(e)$ of $e$ in $f$ is the number of right
bends minus the number of left bends one encounters, when traversing
$e$ from $u$ to~$v$; see Figure~\ref{fig:ortho-rep-1}b.  Note that
every edge has two rotations, one in each face it bounds.  Note
further, that our notation is not precise for bridges, as a bridge is
incident to the same face twice.  However, it will always be clear
from the context which incidence is meant, hence there is no need to
complicate the notation.

Let $uv, vw$ be a path of length~$2$ in the face $f$.  If the two
edges form an angle of $0^\circ$ (i.e., $\rot_f(v) = 2$), the
bend-or-end property of Kandinsky drawings ensures that at least one
of the two edges $uv$ or $vw$ has a bend close to $v$ that forms an
angle of $270^\circ$ in $f$.  To represent this information of which
bends are declared to be close to vertices we introduce some
additional rotations.  Consider the edge $uv$ and let $f$ be an
incident face.  If $uv$ has a bend close to $v$ we define the
\emph{rotation $\rot_f(uv[v])$ at the end} $v$ of $uv$ to be $1$ if it
is a right bend and $-1$ if it is a left bend.  If $uv$ has no bend
close to $v$, we set $\rot_f(uv[v]) = 0$.

It is easy to see, that every Kandinsky representation satisfies the
following properties.  Moreover, it is known that a set of values for
the rotations is a Kandinsky representation if it satisfies these
properties~\cite{fk-ddgn-95} (i.e., there exists a Kandinsky embedding
with these rotation values).
\begin{compactenum}[(1)]
\item \label{item:rotation-of-face}The sum over all rotations in a face
  is~$4$ ($-4$ for the outer face).
\item \label{item:consistent-edges}For every edge $uv$ with incident
  face $f_\ell$ and $f_r$, we have $\rot_{f_\ell}(uv) + \rot_{f_r}(uv)
  = 0$, $\rot_{f_\ell}(uv[u]) + \rot_{f_r}(uv[u]) = 0$, and
  $\rot_{f_\ell}(uv[v]) + \rot_{f_r}(uv[v]) = 0$.
\item \label{item:rotation-around-vertex}The sum of rotations around a
  vertex $v$ is $2\cdot \deg(v) - 4$.
\item \label{item:rotation-range-at-vertex}The rotations at vertices
  lie in the range $[-2, 2]$.
\item \label{item:0-deg-bend-assignment}If $\rot_f(uv, vw) = 2$ then
  $\rot_f(uv[v]) = -1$ or $\rot_f(vw[v]) = -1$.
\end{compactenum}

If the face is clear from the context, we often omit the subscript in
$\rot_f$.  Note that the rotation of an edge $uv$ is split into three
parts; the rotations $\rot(uv[u])$ and $\rot(uv[v])$ at the ends of
$uv$ and a \emph{rotation $\rot(uv[-])$ in the center} of $uv$.  It
holds $\rot(uv) = \rot(uv[u]) + \rot(uv[-]) + \rot(uv[v])$ (thus it is
not necessary to have $\rot(uv[-])$ contained in the representation).
We can assume without loss of generality that all bends accounting for
the rotation in the center bend in the same direction, thus the edge
$uv$ has $|\rot(uv[u])| + |\rot(uv[-])| + |\rot(uv[v])|$ bends.
Hence, the number of bends depend only on the Kandinsky representation
and not on the actual embedding.

Let $f$ be a face of $G$ and let $u$ and $v$ be two vertices on the
boundary of $f$.  By $\pi_f(u,v)$ we denote the path from $u$ to $v$
on the boundary of $f$ in clockwise direction (counter-clockwise for
the outer face).  The rotation $\rot_f(\pi)$ of a path $\pi$ in the
face $f$ is defined as the sum of all rotations of edges and inner
vertices of $\pi$ in $f$.

Note that an orthogonal embedding (of a graph with maximum degree~$4$)
is basically a Kandinsky embedding without $0^\circ$ angles at
vertices.  Thus, we can define \emph{orthogonal representations},
representing an equivalence class of orthogonal embeddings, as
Kandinsky representations where rotation~$2$ at vertices is not
allowed (the resulting notion, although it differs slightly, is
equivalent to the one introduced by Tamassia~\cite{t-emn-87}).

\subsection{Network Flows}

We will need the following result on the existence of feasible flows
in flow networks where the capacity of edges is large compared to the
absolute demands of the nodes in network.

\begin{lemma}
  \label{lem:existence-feasible-flow}
  Let~$N=(V,A)$ be a flow network with demands~$d \colon V \to
  \mathbb{R}$ (with $\sum_{v\in V} d(v) = 0$) and capacities~$c \colon
  A \to \mathbb{R}$ such that $c(e) \ge \sum_{v \in V} |d(v)|$ for
  all~$e \in E$.  Then there exists a feasible flow in~$N$.
\end{lemma}
\begin{proof}
  Let~$\Phi$ be an arbitrary flow satisfying the demands at all
  vertices, but possibly violating the capacity constraints.  Let~$a =
  (u,v) \in A$ with~$\Phi(a) > c(e)$.  If there exists a directed path
  from~$v$ to~$u$ all whose arcs have positive flow, we can decrease
  the amount of flow on this cycle by~1.  After finitely many such
  steps, we then obtain the desired flow.  Hence, assume for the sake
  of contradiction that such a path does not exist.  Let~$S \subseteq
  V$ be the vertices that can be reached from $v$.  Note that~$v \in
  S$ and $u \notin S$.  Hence,~$S$ defines a cut in~$N$ whose outgoing
  arcs have flow~$0$.  In any valid flow the amount of flow entering
  $S$ minus the flow leaving $S$ must equal~$\sum_{v \in S} d(v) \le
  \sum_{v \in S} |d(v)| \le \sum_{v \in V} |d(v)| \le c(e)$.  On the
  other hand, the flow entering $S$ is at least $\Phi(a) > c(e)$ while
  no flow is leaving $S$, a contradiction.
\end{proof}

\section{Complexity}
\label{sec:hardness}

Let $\mathcal S = (\mathcal X, \mathcal C)$ be an instance of
\textsc{3-Sat} with variables $\mathcal X = \{x_1, \dots, x_n\}$ and
clauses $\mathcal C = \{c_1, \dots, c_m\}$.  A clause is a
\emph{positive clause} if it contains only positive literals, a
\emph{negative clause} if it contains only negative literals, and a
\emph{mixed clause} otherwise.  In the \emph{variable-clause graph},
every variable and ever clause is a vertex and there is an edge $xc$
connecting a variable $x \in \mathcal X$ with a clause $c \in \mathcal
C$ if and only if $x \in c$ or $\neg{x} \in c$.

In a \emph{monotone rectilinear representation} of the variable-clause
graph, the variables are represented as horizontal line segments on
the $x$-axis, the positive and negative clauses are represented as
horizontal line segments below and above the $x$-axis, respectively,
and a variable is connected to an adjacent clause by a vertical line
segment such that no two line segments cross.  Note that an instance
admitting a monotone rectilinear representation cannot contain mixed
clauses.  An instance of \textsc{Planar Monotone 3-Sat} is an instance
$S = (\mathcal X, \mathcal C)$ of \textsc{3-Sat} together with a
monotone rectilinear representation of its variable-clause graph; see
Figure~\ref{fig:3-sat-example-1} for an example.  De Berg and
Khosravi~\cite{dk-obsps-12} show that \textsc{Planar Monotone 3-Sat}
is NP-hard.

\begin{figure}
  \centering
  \includegraphics[page=1]{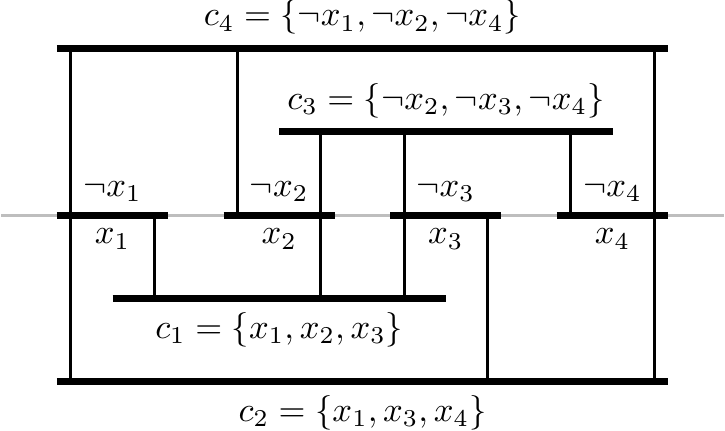}
  \caption{The instance of \textsc{Planar Monotone 3-Sat} with
    variables $x_1, \dots, x_4$ and clauses $\{x_1, x_2, x_3\}$,
    $\{x_1, x_3, x_4\}$, $\{\neg x_2, \neg x_3, \neg x_4\}$, and
    $\{\neg x_1, \neg x_2, \neg x_4\}$.}
  \label{fig:3-sat-example-1}
\end{figure}

The problem \textsc{Orthogonal 01-Embeddability} is defined as
follows.  Given a 4-plane graph $G = (V, E)$ and with partitioned edge
set $E = E_0 \cupdot E_1$, test whether $G$ admits an orthogonal
drawing such that every edge in $E_i$ has exactly $i$ bends.  We also
refer to the edges in~$E_0$ and~$E_1$ as~$0$- and $1$-edges,
respectively.  In the following, we always consider the variant of
\textsc{Orthogonal 01-Embeddability} where we allow to \emph{fix
  angles at vertices}, that is the value of $\rot_f(v)$ for a vertex
$v$ with incident face $f$ might be given with the input.  Fixing the
angles at vertices does not make the problem harder since augmenting a
vertex $v$ to have degree~4 by adding degree-1 vertices incident to
$v$ has the same effect as fixing the angles at $v$ (when choosing the
combinatorial embedding appropriately).  Note that this reduces the
case with fixed angles at vertices to the one without fixed angles.
In the following we implicitly allow angles at vertices to be fixed.

In this section we first show that \textsc{Orthogonal
  01-Embeddability} is NP-hard by a reduction from \textsc{Planar
  Monotone 3-Sat}.  Afterwards, we show that \textsc{Kandinsky Bend
  Minimization} is NP-hard by a reduction from \textsc{Orthogonal
  01-Embeddability}.

\subsection{Orthogonal 01-Embeddability}
\label{sec:orth-01-embedd}

Consider a single $1$-edge $e$.  When drawing it, we have to make the
decision to either bend it in one or the other direction.  In the
reduction from \textsc{Planar Monotone 3-Sat}, this basic decision
will encode the decision to set a variable either to \texttt{true} or
to \texttt{false}.  In addition to that, the construction consists of
several building blocks.  For every variable, we need a gadget that
outputs its positive and its negative literal.  Moreover, we build
gadgets representing clauses that admit a correct drawing if and only
if at least one out of three edges that require one bend is bent in
the desired direction.  Since the same literal usually occurs in
several clauses, we need to copy the decision made for one edge to
several edges.  Finally, we need to bring the decisions of the
variables to the clauses without restricting the possible drawings of
the clauses too much.

In the following we first present some simple gadgets that are used as
building blocks in the following constructions.  Then we start with
the variable gadget that outputs the positive and negative literal of
a variable.  Afterwards, we show how to duplicate literals and then
present the so called bendable pipes that are used to bring the value
of a literal to the clauses.  Finally, we present the clause gadget.
In the end, we put these building blocks together and show the
correctness of the construction.

\subsubsection{Building Blocks}
\label{sec:building-blocks}

An interval gadget is a small graph~$G[\rho_1,\rho_2]$ with two
designated degree-1 vertices (its \emph{endpoints}) $s$ and $t$ on the
outer face.  It has the property that the rotation of~$\pi(s,t)$ is in
the interval~$[\rho_1,\rho_2]$ for any orthogonal embedding.  The
construction is similar to the tendrils used by Garg and
Tamassia~\cite{gt-ccurpt-01}; see Figure~\ref{fig:interval-gadget} for
some examples.

\begin{figure}
  \centering
  \includegraphics[page=1]{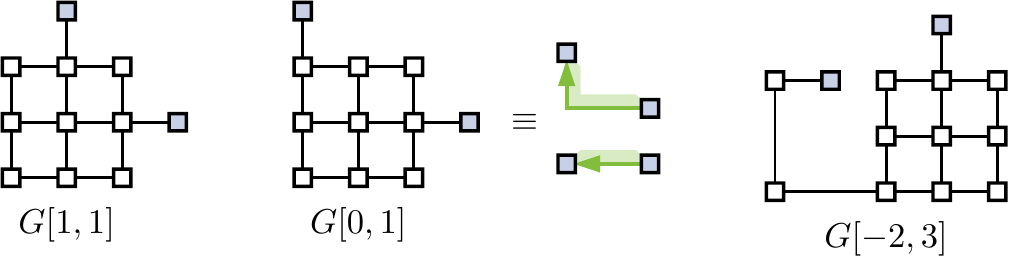}
  \caption{Three different interval gadgets.  The vertices $s$ and $t$
    are marked blue.}
  \label{fig:interval-gadget}
\end{figure}

\begin{lemma}
  \label{lem:interval-gadget}
  The interval gadget $G[\rho_1,\rho_2]$ admits an orthogonal 0-bend
  drawing with rotation $\rho$ if and only if $\rho \in [\rho_1,
  \rho_2]$.
\end{lemma}

The interval gadget we use most frequently in the following is
$G[0,1]$, which behaves like an edge that may have one bend, but only
into a fixed direction (recall that the combinatorial embedding of our
graph is fixed).  To simplify the illustrations, we draw $G[0,1]$ as
shown in Figure~\ref{fig:interval-gadget} and we refer to them
as~$01$-edges.

To simplify the description of the hardness proof, we next describe a
number of basic building blocks, which we combine in different ways to
obtain the gadgets for our construction.  The building blocks are
shown in Figure~\ref{fig:building-blocks}.

Except for the last of the building blocks, each of them consists of a
4-cycle~$s,t,s',t'$.  They only differ in the types of edges.  In the
\emph{box} $st$ and~$s't'$ are $1$-edges and the other edges are
$0$-edges; see Figure~\ref{fig:building-blocks}a.  In a \emph{bendable box} the
two zero-bend edges of a box are replaced by~$01$-edges directed
from~$t$ to~$s'$ and from~$t'$ to~$s$, respectively; see
Figure~\ref{fig:building-blocks}b.  In a \emph{merger} the edge $st$ is
a $1$-edge, $s't'$ and~$st'$ are $01$-edges (with this orientation)
and~$ts'$ is a $0$-edge; see Figure~\ref{fig:building-blocks}c.
Finally, a \emph{splitter} is a 3-cycle~$s,t,s'$, where~$ss'$ is a
$1$-edge and~$s't$ and~$ts$ are $01$-edges (with this orientation);
see Figure~\ref{fig:building-blocks}d.

\begin{figure}[tb]
  \centering
  \includegraphics[page=2]{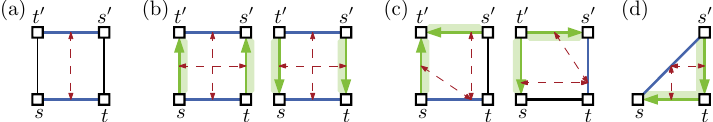}
  \caption{Building blocks for gadgets.  The edges are color-coded;
    $0$-edges are black, $1$-edges are blue and $01$-edges are green
    and directed such that they may bend right but not left.  The
    building blocks are (a)~the box; (b)~the bendable box; (c)~the
    merger; (d)~the splitter.}
  \label{fig:building-blocks}
\end{figure}

Symmetric versions of the bendable box and the splitter can be
obtained by reversing the directions of both~$01$-edges, as shown in
Figure~\ref{fig:building-blocks}b,c.  Since they differ from the
original only by exchanging the inner and outer face and mirroring the
instance, their behavior is completely symmetric.  Note that, apart
from the $0$-edges, all edges of the building blocks admit precisely
two possible rotation values in each face.  Thus, each edge attains
its maximum rotation value in one of its incident faces and the
minimum rotation in the other one.  We call an
orthogonal~$01$-representation of a building block \emph{right-angled}
if all inner angles at vertices are $90^\circ$.  The following lemma
states the functionality of these building blocks, which is
essentially that in a right-angled orthogonal~$01$-representation the
rotation values of some of the edges are not independent of one
another but are linked in the sense that exactly one of them must
attain its minimum (maximum) rotation value in~$f$.  In
Figure~\ref{fig:building-blocks} such dependencies are displayed as
red dashed arrows.  We will later interpret the rotation values as an
encoding of truth values.  The red dashed arrows then correspond to a
transmission of the encoded information.

\begin{lemma}
  \label{lem:building-block-properties}
  Consider a building block~$B$ and assume that we are given rotation
  values for each of the edges incident to the inner face~$f$ of~$B$
  that respect the bend constraints of the edges.  The following
  conditions for each of the building blocks are necessary and
  sufficient for the existence of a right-angled orthogonal
  $01$-representation of~$B$ respecting the given rotation values.
  \begin{compactenum}
  \item Box: Exactly one of~$\{st,s't'\}$ attains its minimum
    (maximum) rotation in~$f$.
  \item Bendable box: Exactly one of~$\{st,s't'\}$ and exactly one
    of~$\{st',ts'\}$ attains its minimum (maximum) rotation in~$f$.
  \item Merger: $st$ attains its minimum (maximum) rotation in~$f$ if
    and only if $st'$ and~$s't'$ attain their maximum (minimum)
    rotation in~$f$.
  \item Splitter: $st$ attains its minimum (maximum) rotation in~$f$
    if and only if~$s't$ and~$t's$ attain their maximum (minimum)
    rotation in~$f$.
  \end{compactenum}
\end{lemma}
\begin{proof}
  We first treat the building blocks that consist of a 4-cycle.
  Denote the rotation values of~$ts,st',t's'$ and~$s't$
  by~$\rho_1,\rho_2,\rho_3,\rho_4$, respectively.  Note that, in any
  valid drawing, each of the vertices contributes a rotation of~1 to
  the inner face~$f$.  Since the total rotation around~$f$ must be~4,
  this implies~$\rho_1 + \rho_2 + \rho_3 + \rho_4 = 0$ is necessary
  and sufficient for the existence of a valid drawing.

  For the box, we have~$\rho_2 = \rho_4 = 0$, and thus~$\rho_1 =
  -\rho_3$ is necessary and sufficient, which implies the claim.

  For the bendable box, observe that~$\rho_2 \in \{0,1\}$ and~$\rho_4
  \in \{-1,0\}$, and thus~$\rho_2 + \rho_4 \in \{-1,0,1\}$.
  Similarly,~$\rho_1, \rho_3 \in \{-1,1\}$, and thus~$\rho_1 + \rho_3
  \in \{-2,0,2\}$.  To achieve a total sum of~$0$, it follows
  that~$\rho_1 + \rho_3 = 0$ and~$\rho_2 + \rho_4=0$ is necessary and
  sufficient.  The claim follows.

  For the merger observe that~$\rho_4=0$.  Moreover, we have~$\rho_2
  \in \{0,1\}$ and~$\rho_3 \in \{-1,0\}$, and thus~$\rho_2 + \rho_3
  \in \{-1,0,1\}$.  Since~$\rho_1 \in \{-1,1\}$, it follows
  that~$\rho_2 + \rho_3 = 0$ can be excluded.  This together with the
  fact that~$\rho_1 = - \rho_2 - \rho_3$ is necessary and sufficient
  proves the claim.

  Finally, we consider the splitter.  We denote the rotations
  of~$ss'$, $s't$ and~$ts$ in~$f$
  by~$\rho_1,\rho_2$ and~$\rho_3$, respectively.  Since each of the
  three vertices incident to~$f$ supplies a rotation of~1, the
  existence of a valid drawing is equivalent to~$\rho_1 + \rho_2 +
  \rho_3 = 1$.  Note that~$\rho_1 \in \{-1,1\}$, whereas~$\rho_2,
  \rho_3 \in \{0,1\}$, and thus~$\rho_2 + \rho_3 \in \{0,1,2\}$.  It
  follows immediately that~$\rho_2 + \rho_3 = 1$ is not possible, and
  thus~$\rho_2 = \rho_3$ is necessary.  Then~$\rho_1 = 1-2\rho_2$
  follows, showing the claim.
\end{proof}

We will now construct our gadgets from these building blocks.  To this
end, we take copies of building blocks and glue them together by
identifying certain edges (together with their endpoints).  As
mentioned above, we will use rotations of the $1$-edges to encode
certain information.  Thus, our gadgets will always have such edges on
the boundary of the outer face.  In the figures, we will again
indicate the necessary conditions from
Lemma~\ref{lem:building-block-properties} by red dashed edges as in
Figure~\ref{fig:building-blocks}.  It follows from
Lemma~\ref{lem:building-block-properties} that when there is a path of
such red edges from one edge to another edge, then they are
synchronized.  In particular, if both are incident to the outer face
than exactly one of them attains the minimum and one of them attains
the maximum rotation there in any valid drawing.

\subsubsection{Gadget Constructions}
\label{sec:gadget-constructions}

\paragraph{Variable Gadget}
\label{sec:variable-gadget}

The variable gadget for a variable~$x$ consists of a single box with
vertices~$s,t,s',t'$.  The two 1-bend edges~$st$ and~$s't'$ are called
the \emph{positive} and \emph{negative output}, respectively.  It
immediately follows from Lemma~\ref{lem:building-block-properties}
that it has exactly two different valid drawings.  We use the
interpretation that~$x$ has value \true if the rotation of the the
positive output in the outer face is maximum, and false otherwise; see
Figure~\ref{fig:variable-gadget}.  The following lemma summarizes the
properties; it follows immediately from
Lemma~\ref{lem:building-block-properties}.

\begin{figure}
  \centering
  \includegraphics[page=1]{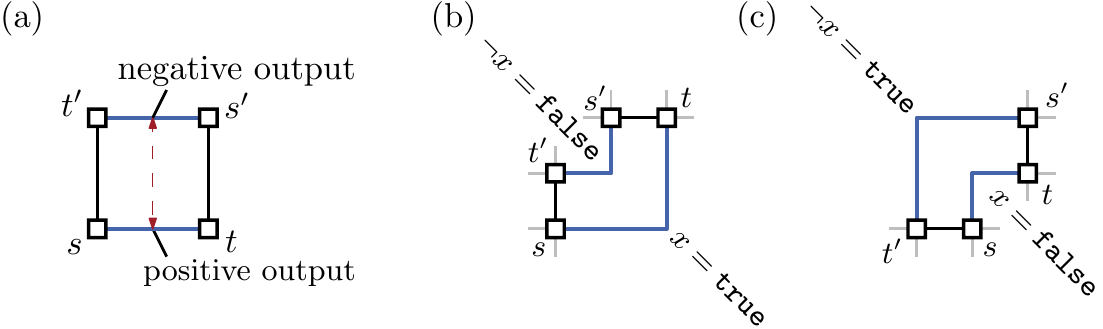}
  \caption{(a)~The variable gadget.  (b--c)~The two possible
    orthogonal representations corresponding to $x = \mathtt{true}$
    and $x = \mathtt{false}$, respectively.}
  \label{fig:variable-gadget}
\end{figure}

\begin{lemma}
  \label{lem:variable-gadget}
  Assume the rotations~$\rho_p$ and~$\rho_n$ of the positive and
  negative output edges in the outer face are fixed.  There is a
  right-angled orthogonal~$01$-representations of the variable gadget
  respecting~$\rho_p$ and~$\rho_n$ if and only if~$\rho_p = -\rho_n
  \in \{-1,1\}$.
\end{lemma}

\paragraph{Literal Duplicator}
\label{sec:literal-duplicator}

A duplicator is a structure that has three $1$-bend edges on the outer
face, one of which is the \emph{input edge}, the other two are the
\emph{output edges}.  The key property is that the structure is such
that the state of the inputs is transferred to both outputs in any
right-angled orthogonal~$01$-representation, i.e., the input attains
its maximum (minimum) rotation in the outer face if and only if the
outputs attains their minimum (maximum) rotation in the outer face.
The duplicator is formed by a splitter, which is glued to two mergers
via its $\{0,1\}$-edges; see Figure~\ref{fig:literal-duplicator}.  The
fact that indeed the information encoded in the input edge is copied
to the output edges follows from the red dashed paths connecting the
input to the outputs and Lemma~\ref{lem:building-block-properties}.

\begin{lemma}
  \label{lem:literal-duplicator-duplicates}
  Assume the rotations~$\rho_i$ of the input edge and the
  rotations~$\rho_o$ and~$\rho_o'$ of the two output edges in the
  outer face are fixed.  There is a right-angled
  orthogonal~$01$-representations of the variable gadget respecting
  these rotations if and only if~$\rho_i = -\rho_o = -\rho_o' \in
  \{-1,1\}$.
\end{lemma}

\begin{figure}
  \centering
  \includegraphics[page=1]{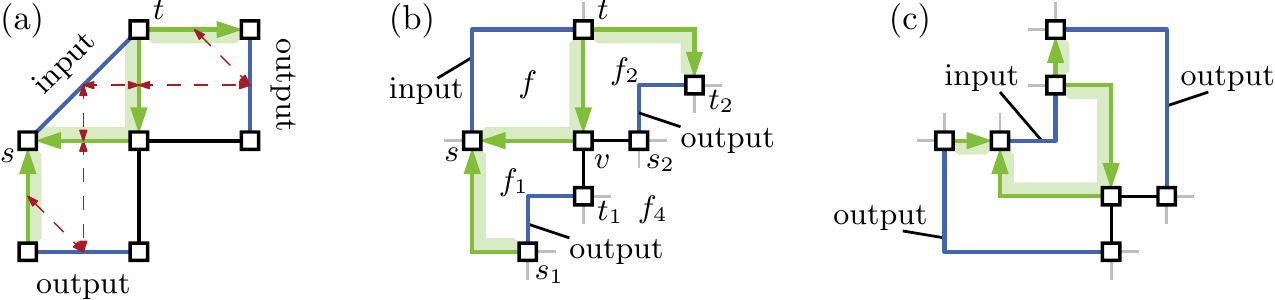}
  \caption{(a) The literal duplicator.  (b--c) The two possible
    orthogonal representations corresponding to the values
    \texttt{false} and \texttt{true}, respectively.}
  \label{fig:literal-duplicator}
\end{figure}

By concatenating several duplicators in a tree-like fashion, we can of
course take as many copies of the state of a literal as there are
clauses containing that literal.  We make this more precise later.

\paragraph{Bendable Pipes}
\label{sec:bendable-pipes}

The bendable pipe gadget is used for transmitting the information
about a literal to a clause.  It has an input and an output edge, and
has the property that in any valid drawing the information encoded in
the input is transmitted to the output.  To remedy the fact that the
duplicators change their shape depending on the state of the literal
they copy, we allow some flexibility of the pipes, allowing them to
change how strongly the pipe is bent.  This is achieved as follows.

A \emph{zig-zag} consists of a bendable box and a bendable box where
the $01$-edges are reversed, such that two of their 1-edges are
identified.  One of the 1-bend edges on the outer face is the input,
the other is the output; see Figure~\ref{fig:zig-zag}.  It follows
immediately from Lemma~\ref{lem:building-block-properties} that the
information from the input is transferred to the output.  Moreover, it
also follows from Lemma~\ref{lem:building-block-properties} that the
decision which of the bendable boxes bend their $01$-edges can be
taken independently.  Thus, the zig-zag allows to choose the
rotation~$\rho,\rho'$ of the paths between the input and the output
edge with~$\rho = - \rho'$ for each~$\rho \in \{-1,0,1\}$.

\begin{figure}[tb]
  \centering
  \includegraphics[page=3]{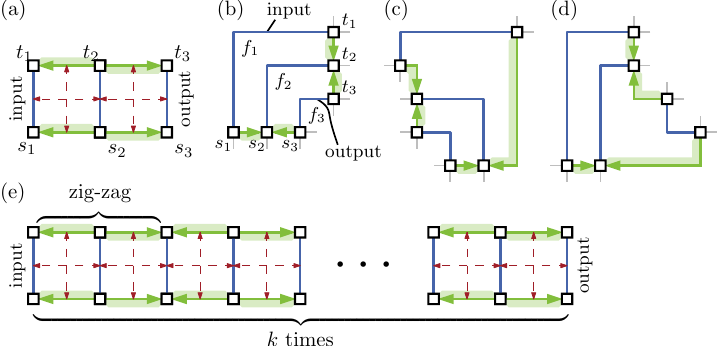}
  \caption{(a)~The zig-zag.  (b--d) Embeddings of the zig-zag with
    different rotations $0$, $1$, and $-1$ when the input edge has
    rotation~$-1$ in the outer face.  Corresponding drawings where the
    rotation of the input edge is~$+1$ are symmetric.  (e)~The
    $k$-bendable pipe.
  }
  \label{fig:zig-zag}
\end{figure}

A \emph{$k$-bendable pipe} is obtained by concatenating~$k$ zig-zags;
see Figure~\ref{fig:zig-zag}e.  Again
Lemma~\ref{lem:building-block-properties} easily implies that the
information is transmitted from the input to the output, and moreover,
by concatenating suitable drawings of the zig-zags, for each
rotation~$\rho \in \{-k,\dots,k\}$, the paths between the input and
the output edge along the outer face can have rotation~$\rho$
and~$-\rho$, respectively.  In a high-level view, a $k$-bendable pipe
looks like an edge that transfers information between its endpoints
and can be bent up to $k$ times either to the left or to the right.
The following lemma summarizes the properties of~$k$-bendable pipes.

\begin{lemma}
  \label{lem:1-bendable-pipe-is-pipe}
  Assume the rotations~$\rho_i$ and~$\rho_o$ of the input edge and the
  output edge as well as the rotations~$\rho$ and~$\rho'$ of the two
  counterclockwise paths on the outer face connecting the input and
  the output edge are fixed.

  There is a right-angled orthogonal~$01$-representations of the
  $k$-bendable pipe if and only if~$\rho_i = -\rho_o \in \{-1,1\}$
  and~$\rho = -\rho' \in \{-k,\dots,k\}$.
\end{lemma}

\paragraph{Clause Gadget}
\label{sec:clause-gadget}

The \emph{clause gadget} is a cycle $C$ of length~4, consisting of
three~$1$-edges, the \emph{input edges}, and the interval gadget
$G[-2,3]$; see Figure~\ref{fig:clause-gadget}a.  The embedding is
fixed such that the inner face of the clause lies to the right of the
interval gadget $G[-2,3]$ (that is the rotation of $G[-2,3]$ in the
inner face lies in the interval $[-2,3]$).  Again we only consider
right-angled drawings, where the rotations at the vertices in the
internal face are all fixed to~$1$.

The clause gadget interprets a rotation of $-1$ for an input edge in
the inner face as \texttt{true} and a rotation of $1$ as
\texttt{false}.  In Figure~\ref{fig:clause-gadget}a all three input
edges are set to \texttt{true}.  In Figure~\ref{fig:clause-gadget}b
two of the three input edges represent the value \texttt{false}.  In
Figure~\ref{fig:clause-gadget}c all input edges are \texttt{false},
thus $G[-2,3]$ would need to have a rotation of $-3$ in the inner
face, which is not possible.  The following lemma states more
precisely that the clause gadget admits a valid drawing with the given
rotations of the input edges if and only if at least one of the input
edges represents the value \texttt{true}.

\begin{figure}
  \centering
  \includegraphics[page=1]{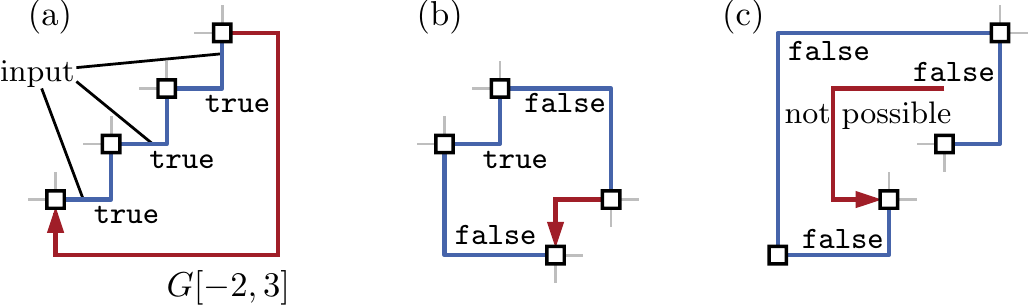}
  \caption{(a--c) The clause gadget with three different values on the
    input edges.}
  \label{fig:clause-gadget}
\end{figure}

\begin{lemma}
  \label{lem:clause-gadget}
  Assume that the rotation~$\rho_1,\rho_2,\rho_3$ of the input edges
  in the inner face are fixed.  There exists a right-angled orthogonal
  01-representation of the clause gadget respecting these rotations if
  and only if~$\rho_i \in \{-1,1\}$ for~$i \in \{1,2,3\}$ and~$\rho_i =
-1$ for at least one~$i \in \{1,2,3\}$.
\end{lemma}
\begin{proof}
  Each of the four vertices of the clause gadget $C$ has rotation~$1$
  in the inner face of $C$.  Thus the sum of the rotations $\rho_1$,
  $\rho_2$, $\rho_3$, and the rotation of $G[-2,3]$ in the inner face
  of $C$ must be~$0$.  The possible rotation of $G[-2,3]$ are exactly
  the integers in the interval $[-2, 3]$
  (Lemma~\ref{lem:interval-gadget}).  Thus, we get an orthogonal
  01-representation if and only if $\rho_1 + \rho_2 + \rho_3 \in [-3,
  2]$, which is the case if and only if not all three rotations
  are~$1$.
\end{proof}

\subsubsection{Putting Things Together}

Let $S = (\mathcal X, \mathcal C)$ together with a monotone
rectilinear representation be an instance of \textsc{Planar Monotone
  3-Sat}.  The plan is to create a variable gadget for every variable
and a clause gadget for every clause, duplicate the literals (using
the literal duplicator) outputted by the variable gadget as many times
as they occur in clauses, and bring the values of the duplicated
literals to the input of the clauses using bendable pipes.

Thus, if we have two gadgets $A$ and $B$, we want to use an output
edge of $A$ as the input edge of $B$.  To make the description
simpler, we assume each input edge and each output edge of the gadgets
to be oriented such that the outer face lies to its left and to its
right, respectively.  We can combine $A$ and $B$ by identifying an
output edge $e_A$ of $A$ with an input edge $e_B$ of $B$ such that
their sources and targets coincide.  All input and output edges of the
two gadgets remain input and output edges in the resulting graph,
except for $e_A$ and $e_B$.

Let $x \in \mathcal X$ be a variable.  We take one variable gadget $X$
representing the decision made for $x$.  Let $k$ be the number of
clauses containing the literal $x$.  We successively add $k-1$ literal
duplicators.  The input edge of the first literal duplicator is
identified with the positive output edge of $X$.  The input edge of
every following literal duplicator is identified with an output edge
of a previously added literal duplicator.  The graph we get has the
negative output edge at $X$ and $k$ output edges belonging to literal
duplicators.  To each of these $k$ output edges we add a $K$-bendable
pipe for a suitably large~$K$ by identifying the output edge with the
input edge of the bendable pipe.  We choose $K=3m^2 + 4m$, where $m$
is the number of edges in the variable-clause graph of $S$.  Let $k'$
be the number of clauses containing the literal $\neg x$.  As for the
positive literal, we add $k'-1$ literal duplicators, this time
identifying the input edge of the first literal duplicator with the
negative output edge of $X$.  As before, we also add $K$-bendable
pipes to each of the $k'$ output edges.  We call the resulting graph
\emph{variable tree of $x$} and denote it by $T_x$.  We call the $k$
output edges of the bendable pipes attached to literal duplicators
attached to the positive output edge of the variable gadget $X$ the
\emph{positive output edges} of $T_x$.  The $k'$ other output edges
are \emph{negative output edges} of $T_x$.  The variable tree for the
case $k = 5$ and $k' = 2$ is illustrated in
Figure~\ref{fig:variable-tree}.

\begin{figure}
  \centering
  \includegraphics[page=2]{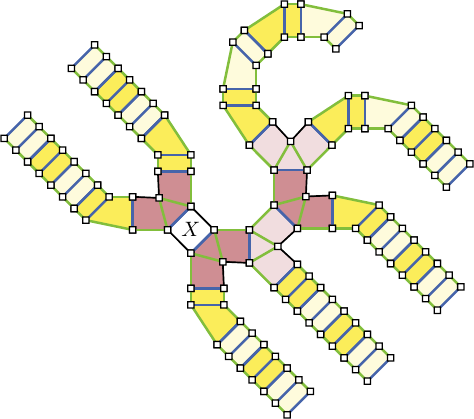}
  \caption{Variable tree~$T_x$ of a variable~$x$ whose positive and
    negative literal have five and two occurrences, respectively.  The
    variable gadget is shaded white, duplicators are shaded in red and
    zig-zags (forming bendable pipes) are shaded yellow.  Adjacent
    gadgets of the same type are shaded with different saturations.}
  \label{fig:variable-tree}
\end{figure}

For the instance $S = (\mathcal X, \mathcal C)$ of \textsc{Planar
  Monotone 3-Sat} we create the following instance of
\textsc{Orthogonal 01-Embeddability}.  For every variable $x \in
\mathcal X$, we take the variable tree $T_x$.  For every clause $c \in
\mathcal C$, we add a copy of the clause gadget.  We connect them by
identifying the output edges of the variable trees with the input
edges of the clause gadget in the following way.  

Consider a variable $x$ and a positive clause $c$ with $x \in c$ in
the monotone rectilinear representation of $S$.  We say that $c$ is
the \emph{$i$th positive clause} of $x$ if the edge connecting $c$ and
$x$ is the $i$th edge incident to $x$ (ordered from left to right).
Analogously, $x$ is the \emph{$j$th variable} of $c$ if this edge is
the $j$th edge incident to $c$.  In the instance shown in
Figure~\ref{fig:3-sat-example-1} and Figure~\ref{fig:3-sat-example-2},
the clause $c_1$ is the first positive clause of $x_2$ and $x_2$ is
the second variable of $c_1$.  Analogously, we define the \emph{$i$th
  negative clause}.

Let $c$ be the $i$th positive clause of $x$ and let $x$ be the $j$th
variable of $c$.  Let further $C$ be the clause gadget corresponding
to $c$.  Traversing the outer face of $C$ in counter-clockwise order
starting with the interval gadget defines an order on the input edges
of $C$.  Moreover, traversing the variable tree $T_x$ in
counter-clockwise order starting with an edge incident to the variable
gadget defines an order on the positive output edges of $T_x$.  We
identify the $i$th positive output edge of $T_x$ with the $j$th input
edge of $C$.  For a negative clause containing $\neg x$, we do exactly
the same except for defining the order of the negative output edges by
traversing the outer face of $T_x$ in clockwise order.  This
identification of input with output edges is done for every edge in
the variable-clause graph.  We denote the resulting graph by $G(S)$.
Figure~\ref{fig:3-sat-example-2} shows the monotone rectilinear
representation (rotated by $45^\circ$) of an example instance $S$ and
the graph $G(S)$.  The graph $G(S)$ has two kinds of faces.  Faces that
are inner faces in the variable tree or in the clause gadget are
called \emph{small faces}.  The other faces are \emph{large faces}.
Note that there is a one-to-one correspondence between the large faces
of $G(S)$ and the faces of the variable-clause graph of $S$.
We obtain the following theorem by proving that $S$ admits a
satisfying truth assignment if and only if $G(S)$ admits an orthogonal
01-representation.

\begin{figure}[t]
  \centering
  \includegraphics[width=1\linewidth,page=2]{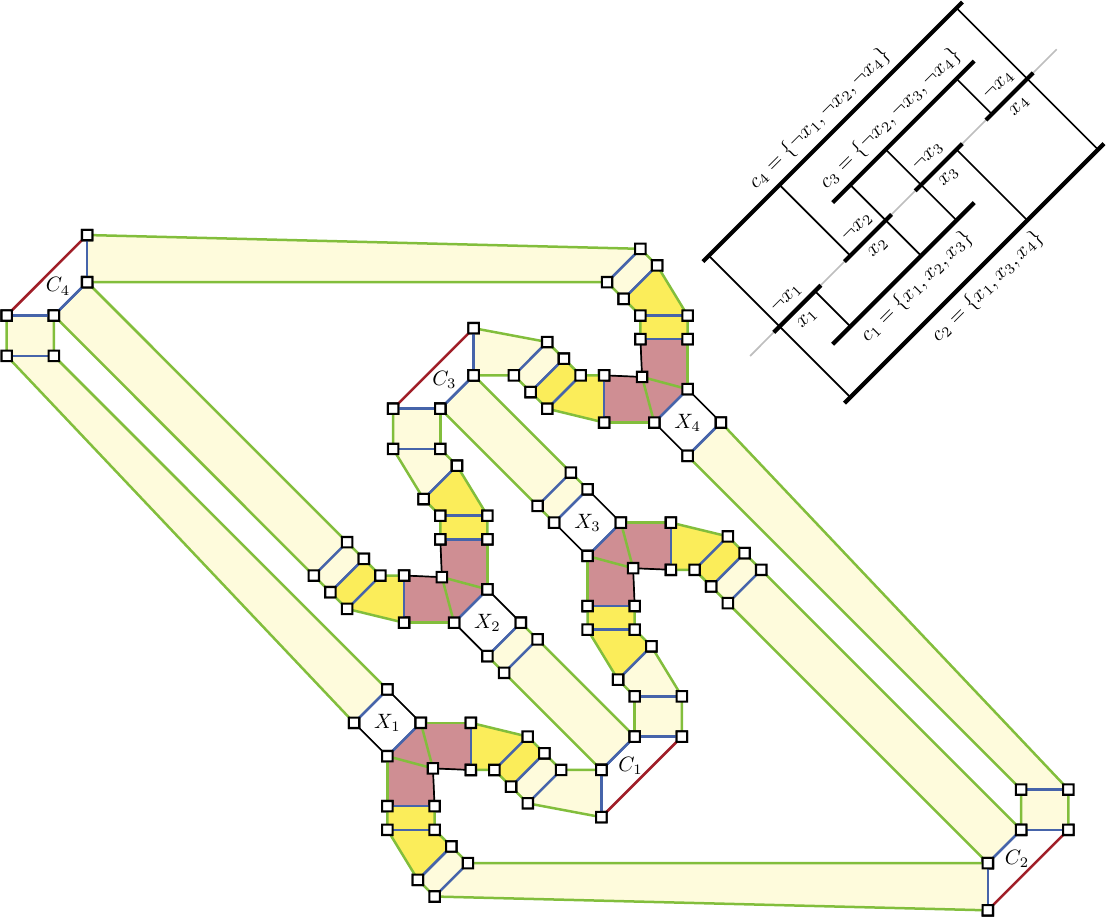}
  \caption{Example reduction of {\sc Planar Monotone 3-Sat} to {\sc
      Orthogonal 01-Embeddability}.  The bendable pipes have been
    shortened for clarity.}
  \label{fig:3-sat-example-2}
\end{figure}

\begin{theorem}
  \label{thm:np-hardness}
  \textsc{Orthogonal 01-Embeddability} is NP-complete.
\end{theorem}
\begin{proof}
  Let $S = (\mathcal X, \mathcal C)$ be an instance of
  \textsc{Monotone Planar 3-Sat} and let $G(S)$ be the graph
  constructed from $S$ as defined above.  We first show that the
  existence of an orthogonal 01-representation of $G(S)$ implies the
  existence of a satisfying truth assignment for $S$.

  Let $\mathcal O$ be an orthogonal 01-representation of $G(S)$.  Let
  $x \in \mathcal X$ be a variable and let $X$ be the corresponding
  variable gadget in $G(S)$.  If the positive output edge of $X$ has
  rotation $-1$ in its outer face, we set $x = \mathtt{true}$ (as
  illustrated in Figure~\ref{fig:variable-gadget}b).  Otherwise, we
  set $x = \mathtt{false}$ (as illustrated in
  Figure~\ref{fig:variable-gadget}c).  We claim that this gives a
  satisfying truth assignment for $S$.  Let $c \in \mathcal C$ be a
  positive clause and let $C$ be the corresponding clause gadget in
  $G(S)$.  By Lemma~\ref{lem:clause-gadget}, at least one of the input
  edges of $C$ has rotation~$-1$ in its inner face.  By construction
  of $G(S)$, this input edge is identified with an positive output
  edge of the variable tree $T_x$ for a variable $x$.  Let $X$ be the
  corresponding variable gadget.  As there is a path of literal
  duplicators and bendable pipes from the positive output edge of $X$
  to every positive output edge of $T_x$, it follows from
  Lemma~\ref{lem:literal-duplicator-duplicates} and
  Lemma~\ref{lem:1-bendable-pipe-is-pipe} that the positive output
  edge of $X$ has rotation~$-1$ in its outer face if and only if any
  positive output edge of $T_x$ has rotation~$-1$ in the outer face of
  $T_x$.  Thus, it follows that the positive output edge of $X$ has
  rotation~$-1$ in its outer face and thus $x = \mathtt{true}$, which
  satisfies the clause $c$.

  If $c$ is a negative clause, we find a variable $x$ such that the
  negative output edge of the corresponding variable gadget $X$ has
  rotation~$-1$ in its outer face.  By
  Lemma~\ref{lem:variable-gadget}, the positive output edge of $X$ has
  rotation~$1$ in its outer face, thus $x = \mathtt{false}$ holds,
  which satisfies the negative clause $c$ containing $\neg x$.

  It remains to show the opposite direction.  Assume we have a
  satisfying truth assignment for $S$.  We show how to construct an
  orthogonal 01-representation of~$G(S)$.  As $G(S)$ consists of
  gadgets for which the rotations around every vertex are fixed, it
  remains to specify a rotation for every edge such that the rotation
  around every inner face is~$4$.  We start with the small faces.
  Consider the variable tree $T_x$ of a variable $x$ containing the
  variable gadget $X$.  If $x = \mathtt{true}$, we choose the
  orthogonal 01-representation of $X$ where the positive output edge
  has rotation~$-1$.  This yields a feasible representation by
  Lemma~\ref{lem:variable-gadget}; see
  Figure~\ref{fig:variable-gadget}b--c.

  This already fixes the rotation of the literal duplicators in $T_x$
  that are directly attached to the output edges of $X$.  By
  Lemma~\ref{lem:literal-duplicator-duplicates} this fixes the
  rotation of the corresponding output edge (to the same behavior as
  the input edge) and a corresponding orthogonal $01$-representation
  of the duplicator exists; see
  Figure~\ref{fig:literal-duplicator}b--c.  Applying this procedure
  iteratively to every literal duplicator whose input edge has a fixed
  rotation fixes the orthogonal representation of every literal
  duplicator in $T_x$.

  Similarly, we (partially) fix the orthogonal 01-representation of
  the bendable pipes contained in $T_x$ iteratively according to
  Lemma~\ref{lem:1-bendable-pipe-is-pipe}.  More precisely, the
  rotation of the 1-edges is fixed according to the rotation of the
  input edge; see Figure~\ref{fig:zig-zag}e.  However, we do not
  fix the rotation of the bendable pipes.  Recall that, by
  Lemma~\ref{lem:1-bendable-pipe-is-pipe} this rotation can be
  anything in $\{-K,\dots,K\}$.  We will need the flexibility of
  choosing this rotation to get the rotations in the large faces
  right.

  Note that the resulting orthogonal 01-representations of the
  variable tree have the following properties.  The positive output
  edges of $T_x$ have rotation~$-1$ if $x = \mathtt{true}$ and
  rotation~$1$ otherwise.  The negative output edges have
  rotation~$-1$ if $\neg x = \mathtt{true}$ and rotation~$1$
  otherwise.  By fixing the orthogonal representations of the variable
  trees in this way, we already fix the orthogonal representation of
  the input edges of the clause gadgets in $G(S)$.  Let $C$ be a
  clause gadget in $G(S)$.  Since $S$ is a satisfying truth
  assignment, it follows that the rotation of at least one input edge
  of $C$ in the inner face of $C$ is~$-1$.  Thus, $C$ admits an
  orthogonal 01-representation by Lemma~\ref{lem:clause-gadget}.  The
  choices made so far imply that every small face in our orthogonal
  01-representation has rotation~4, as required.

  It remains to choose the rotations of the bendable pipes such that
  the rotation in the large inner faces is~$4$.  Initially, assume
  that the rotation of every bendable pipe is~$0$.  We first bound the
  maximum deviation from a rotation of~4 around large faces.

  Let $f$ be a large face and let $f_S$ be the corresponding face in
  the variable-clause graph of $S$.  The boundary of $f$ can be
  naturally subdivided into paths belonging to different variable
  trees and paths on the outer face of clause gadgets.  Let $x$ be a
  variable on the boundary of $f_S$.  A path between two output edges
  of $T_x$ consists of three subpaths.  Two paths with rotation~$0$
  consisting of edges belonging to bendable pipes and, in between, one
  path of edges belonging to literal duplicators.  Clearly, this path
  has length at most $\deg(x)$ and since the absolute value of the
  rotation at edges and vertices is at most~$1$, we get a total
  rotation between $-2\deg(x)$ and~$2\deg(x)$ in the large face.
  Summing over all variables incident to $f_S$ gives us a rotation
  between~$-2m$ and~$2m$, where $m$ is the number of edges in the
  variable-clause graph.  Moreover, for each clause incident to $f_S$
  the boundary of $f$ contains a path having absolute rotation at
  most~$3$.  As there are $m/3$ clauses, the total rotation around the
  large face $f$ is between $-3m$ and~$3m$.

  Changing the rotation of a bendable pipe increases the rotation of
  one incident large face by~$1$ and decreases it in the other
  incident large face by~$-1$.  (Note that this does not affect the
  rotations at small faces.)  Thus, choosing the rotations of the
  bendable pipes such that the rotation in every large face is~$4$
  (except for the outer face with rotation~$-4$) is equivalent to
  finding a flow in the flow network $N$ defined as follows.  The
  underlying graph of $N$ is the dual graph of the variable-clause
  graph of $S$.  The demand of the node corresponding to the face
  $f_S$ is the difference between the rotation in the corresponding
  large face $f$ of $G(S)$ and~$4$ ($-4$ if $f$ is the outer face).
  Note that the demands sum up to~$0$.  The capacity on an edge
  connecting $f_S$ and $f_S'$ is equal to the total length of the
  bendable pipes incident to the corresponding faces $f$ and $f'$ in
  $G(S)$, and thus at least~$K$.  As shown above, the absolute value
  of the demand of each node in the flow network is at most $3m + 4$.
  As the flow network contains at most $m$ nodes (otherwise it would
  be a tree or disconnected), the sum of the absolute values of the
  demands is bounded by $3m^2 + 4m$.  The capacity of every edge in
  $N$ is at least $K=3m^2 + 4m$ by the construction of the variable
  tree.  By Lemma~\ref{lem:existence-feasible-flow} the network $N$
  has a solution.
\end{proof}

\begin{theorem}
  \label{thm:orthogonal-01-emb-hard-variants}
  \textsc{Orthogonal 01-Embeddability} is NP-hard for all combinations
  of the following variations.
  \begin{compactitem}
  \item The input has a fixed planar embedding or a fixed planar
    embedding up to the choice of an outer face.
  \item The angles at vertices incident to~$1$-edges are fixed or
    variable, while angles at vertices incident to $0$-edges are
    variable.
  \end{compactitem}
\end{theorem}
\begin{proof}
  In the construction showing Theorem~\ref{thm:np-hardness}, we
  already fixed all angles at vertices incident to 1-edges (the only
  vertices whose angles are not fixed lie inside interval gadgets).
  Thus, we already established hardness for the case that all angles
  at vertices incident to 1-edges are fixed.  As mentioned before,
  fixing angles is not a really a restriction, as we can enforce fixed
  angles by attaching degree-1 vertices.

  It remains to show that the problem remains hard when allowing to
  choose a different outer face.  Clearly, when choosing a different
  large face as outer face all arguments leading to a satisfying truth
  assignment remain valid.  Moreover, choosing a small face as outer
  face can never lead to a valid orthogonal 01-representation for the
  following reason.  Each small face is one of the building blocks
  presented in Section~\ref{sec:building-blocks} (see
  Figure~\ref{fig:building-blocks}), or the inner face of a clause
  gadget (Figure~\ref{fig:clause-gadget}).  For the building blocks it
  is easy to see that the total rotation in the inner face is at
  least~$0$ (by the fixed angles and the restriction of bends on the
  edges).  Thus, none of them can be chosen as the outer face (which
  would require a rotation of~$-4$).  Similarly, the rotations at
  every vertex in the clause gadget is~$1$ in its inner face, which
  sums up to a rotation of~$4$.  The three input edges have rotation
  at least~$-1$ in the inner face and the interval gadget has rotation
  at last~$-2$.  Thus, the total rotation is at least~$-1$, which
  makes it impossible to choose it as the outer face.
\end{proof}

By the equivalence of orthogonal representations to flow
networks~\cite{t-emn-87}, it follows that it is NP-hard to test
whether there is a valid flow in a planar flow network with the
properties that
\begin{inparaenum}[(i)]
\item the capacity on every edge is~1 and
\item some undirected edges require to have one unit of flow (no
  matter in which direction).
\end{inparaenum}
Note that Garg and Tamassia~\cite{gt-ccurpt-01} show hardness for the
less restrictive case that the capacities and the lower bounds for
flow on undirected edges is unbounded.  They use this to show
NP-hardness of \textsc{Orthogonal 0-Embeddability} of 4-planar graph
(with variable combinatorial embedding).

\begin{theorem}
  All variants of \textsc{Orthogonal 01-Embeddability} are NP-hard
  even if the input graph is a subdivision of a 3-connected graph.
\end{theorem}

\begin{proof}
  We reduce from \textsc{Orthogonal 01-Embeddability} with fixed
  planar embedding and variable angles.  Let~$G=(V,E_0 \cupdot E_1)$ be a
  connected instance of this problem.  We replace each degree-1
  vertex~$v$ by a cycle~$C$ of four 0-edges such that one vertex
  of~$C$ is adjacent to the neighbor of~$v$.  It is not hard to see
  that the resulting graph has an orthogonal~01-embedding if and only
  if~$G$ has one.  In the following we assume without loss of
  generality that~$G$ has minimum degree~2.

  For each vertex~$v$ with incident edges~$e_1,\dots,e_d$ (in
  clockwise order around~$v$), we make the following construction.
  First, we subdivide its incident edges~$e_i$ with new vertices~$v_i$
  and connect them to form a cycle (in the clockwise ordering
  around~$v$), and subdivide the edges of this cycle five times.  The
  vertices before and after~$v_i$ in clockwise direction are
  denoted~$v_i^-$ and~$v_i^+$.  Afterwards, each edge~$uv$ has been
  subdivided into~$uu_i,u_iv_j,v_jv$.  We now add for each such edge
  the edges~$u_i^-v_j^+$ and~$u_i^+v_j^-$.  The edges~$u_iv_j$,
  $u_i^-v_j^+$, and~$u_i^+v_j^-$ are 1-edges if and only if the
  original edge~$uv$ was a 1-edge.  All other edges are 0-edges.
  Figure~\ref{fig:01-triconnected}a illustrates the construction for
  a vertex of degree~3.

  \begin{figure}[tb]
    \centering
    \includegraphics[page=4]{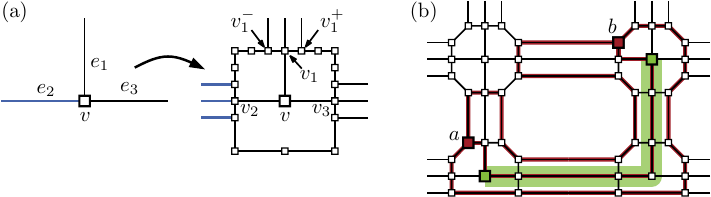}
    \caption{(a)~Construction for transforming an instance of
      \textsc{Orthogonal $01$-Embeddability} into a subdivision of a
      3-connected graph.  A vertex $v$ of degree~3 with variable angles
      (left) and the corresponding gadget for the construction
      (rights).  
      (b)~The routing of three disjoint paths from $a$ to $b$ in $G'$;
      subdivision vertices are omitted.  Vertices~$u$ and $v$ are
      marked green, and the corresponding path in $G$ is bold green.
      The three red paths between $a$ and $b$ follow the bold green
      path.  Note that at the beginning at the end of a path some
      rerouting via vertices not on the path may be necessary,
      however, the rerouting is such that the paths remain disjoint.
    }
    \label{fig:01-triconnected}
  \end{figure}

  We claim that the resulting graph~$G'$
  \begin{inparaenum}[(i)]
  \item admits an orthogonal 01-represen\-tation if and only if~$G$ does, and
  \item is a subdivision of a 3-connected graph.
  \end{inparaenum}
  Once the claim is proved, the statement of the theorem follows since
  the reduction can be performed in polynomial time.
  
  We start with~(i).  First assume that~$G'$ has an orthogonal
  01-representation~$\mathcal O$ and let~$uv$ be an edge of~$G$ that
  is subdivided into~$u_iv_j$.  By construction both~$u_i$ and~$v_j$
  have degree~4 and the edges~$uu_i$ and~$v_jv$ are~$0$-edges.  That
  is all bends of the path~$uu_iv_jv$ lie on the edge~$u_iv_j$.
  Hence, the representation on the subgraph containing the
  vertices~$\{v,v_1,\dots,v_{\deg(v)} \mid v \in V\}$ has all bends on
  the edges~$v_iv_j$.  We can then undo the subdivisions and obtain an
  orthogonal 01-representation of~$G$.  Conversely, if~$\mathcal O$ is
  an orthogonal 01-representation of~$G$, we can first subdivide each
  edge~$uv$ close to vertices~$u$ and~$v$ to obtain vertices~$u_i$
  and~$v_j$ with~$uu_i$ and~$v_jv$ having~$0$ bends.  Then we add the
  edges the edges~$u_i^-v_j^+$ and~$u_i^+,v_j^-$ parallel to~$u_iv_j$.
  Finally, we add the remaining edges of the cycle around~$v$, which
  can be done without bends on the edges since the paths from~$v_i^+$
  to~$v_{i+1}^-$ (indices taken modulo~$\deg(v)$) have sufficiently
  many degree-2 vertices, which can serve as bends.  We have obtained
  an orthogonal 01-representation of~$G'$.

  For~(ii), we show that in~$G'$ any two vertices~$a$ and~$b$ of
  degree~3 or more are connected by three (internally) vertex-disjoint
  paths.  Let~$u$ and~$v$ be the two vertices of~$G$ to whose
  construction~$a$ and~$b$ belong.  If~$u=v$ it is not hard to find
  three disjoint paths; one path goes through the center vertex~$v$,
  the remaining paths are routed clockwise and counterclockwise along
  the cycle around~$v$.  It may be necessary to route through a
  neighboring gadget to get around the attachment vertices of~$v$; see
  Figure~\ref{fig:01-triconnected}b.

  If~$u \ne v$, we pick a shortest path~$u = u_1,\dots,u_k = v$
  from~$u$ to~$v$ in~$G$.  This path corresponds to a
  path~$u_1a_1b_2u_2,\dots,a_{k-1}b_ku_k$, where~$a_i$ and~$b_i$ are
  vertices of the cycle around vertex~$u_i$.  We find three disjoint
  paths from~$a_1^-, a_1$ and~$a_1^+$ to~$b_k^-, b_k$ and~$b_k^+$,
  respectively, simply by taking for each edge~$u_iu_{i+1}$ for~$1 < i
  < k-1$ the path from~$a_i^+$ in clockwise direction along the cycle
  around~$u_i$ via~$b_i^-$ to~$a_{i+1}^+$, from~$a_i$ via the center
  vertex~$u_i$ and~$b_i$ to~$a_{i+1}$, and from~$a_i^+$ in
  counterclockwise direction along the cycle around~$u_i$ via~$b_i^+$
  to~$a_{i+1}^-$; this is illustrated in the middle vertex of the
  green path in Figure~\ref{fig:01-triconnected}b.  We also extend
  these paths by adding edges~$a_{k-1}b_k$, $a_{k-1}^+b_k^-$
  and~$a_{k-1}^-b_k^+$ so that we have disjoint paths from~$a_1^-$
  to~$b_k^+$, from~$a_1$ to~$b_k$ and from~$a_1^+$ to~$b_k^-$.  It
  then remains to find disjoint paths from~$a$ to~$a_1^-,a_1$
  and~$a_1^+$ and from~$b_k^-,b_k$ and~$b_k^+$ to~$b$.  This can be
  done by routing in the gadget around~$u$ and~$v$, respectively.
  Note that it may be necessary to visit the gadget of an adjacent
  vertex.  This does, however, not interfere with the paths
  constructed so far since we assumed that it is a shortest path, and
  hence the corresponding neighbors are not part of the constructed
  paths.  This finishes the proof of the claim.
\end{proof}

\begin{corollary}
  \textsc{Orthogonal 0-Embeddability} is NP-hard for 4-planar graphs
  with a variable combinatorial embedding.
\end{corollary}

\begin{proof}
  We reduce from {\sc Orthogonal 01-Embeddability} where the input
  graph is a subdivision of a 3-connected graph.  Note that the
  embedding is unique up to the choice of the outer face.  We now
  replace each~$1$-edge by a copy of the interval gadget~$G[1,1]$ (see
  Figure~\ref{fig:interval-gadget}).  Changing the embedding of this
  gadget decides the bend direction of the $1$-edge and vice versa.
  It is not hard to see that the resulting graph admits
  a~$0$-embedding if and only if the original instance admits an
  orthogonal~$01$-embedding.  Clearly the reduction runs in polynomial
  time.
\end{proof}

\subsection{Kandinsky Bend Minimization}
\label{sec:kand-bend-minim}

In the following, we show how to reduce \textsc{Orthogonal
  01-Embeddability} to \textsc{Kandinsky Bend Minimization}.  The
reduction consists of two basic building blocks.  In an orthogonal
embedding, every side of a vertex can be occupied by at most one edge.
We show how to enforce this requirement also for Kandinsky embeddings.
Moreover, we construct a subgraph whose Kandinsky embeddings behave
like the embeddings of an edge with exactly one bend.

\paragraph{Corner Blocker}

Let $B$ be the graph consisting of a 4-cycle together with an
additional \emph{attachment vertex} connected to two non-adjacent
vertices of the 4-cycle.  The graph $B$ is called \emph{corner
  blocker}.  Figure~\ref{fig:corner-blocker}a shows a corner blocker
with attachment vertex $v$.  Let $v$ be a vertex in a planar graph
$G$.  \emph{Blocking a corner of $v$} denotes the process of attaching
a corner blocker to $v$ by identifying the attachment vertex of $B$
with $v$.  Consider a Kandinsky embedding of $G + B$.  The corners of
the box representing the vertex $v$ are also called the \emph{corners}
of $v$.  We say that a corner of $v$ is \emph{blocked} by the corner
blocker $B$ if it lies in the inner face of $B$ incident to $v$.
Figure~\ref{fig:corner-blocker}b shows a vertex $v$ with four corner
blockers attached to it such that all four corners of $v$ are blocked.
Note that the Kandinsky representation of a Kandinsky embedding
already determines which corners are blocked by a corner blocker.

\begin{figure}
  \centering
  \includegraphics[page=1]{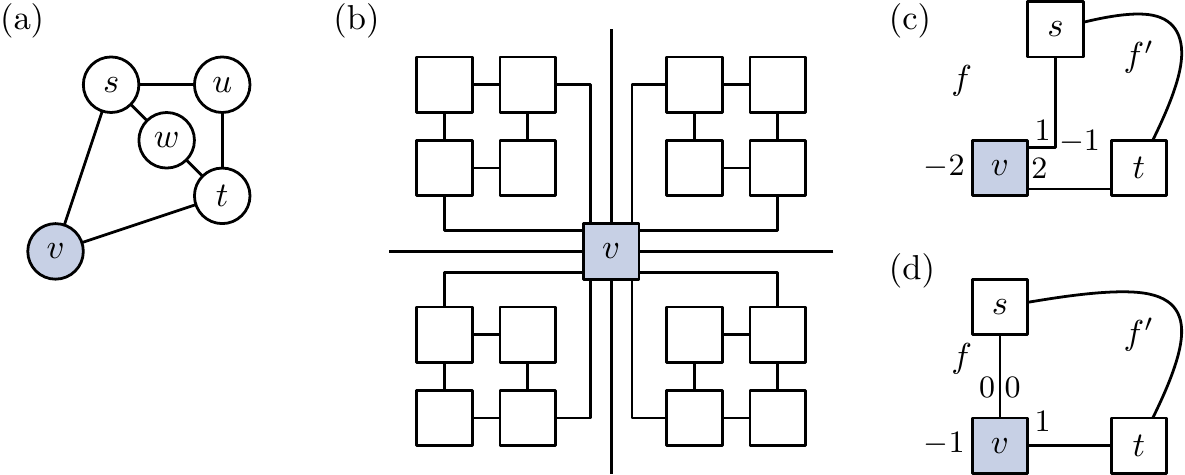}
  \caption{(a)~A corner blocker with attachment vertex $v$.  (b)~A
    Kandinsky representation of a vertex $v$ with four attached corner
    blockers (and four outgoing edges).  (c--d) Illustration of the
    proof of Lemma~\ref{lem:corner-blocker-lower-bound}.}
  \label{fig:corner-blocker}
\end{figure}

The idea behind the corner blocker is to enforce a blocking of all
four corners of a vertex.  Recall that we assume a fixed planar
embedding of the input graph and thus a fixed order of edges around
every vertex.  Thus, blocking all four corners is equivalent to
enforcing edges to leave a vertex at a specific side, as in
Figure~\ref{fig:corner-blocker}b.  The following two lemmas show
that the corner blocker defined above is well suited for this purpose,
as it admits an optimal drawing blocking only a single corner but
blocking no corner causes additional cost.

\begin{lemma}
  \label{lem:corner-blocker-lower-bound}
  Every Kandinsky representation of a corner blocker has at least two
  bends.
\end{lemma}
\begin{proof}
  Let $B$ be a corner blocker.  Denote the degree-2 vertices of $B$
  with $u$, $v$, and $w$ and the degree-3 vertices with $s$ and $t$
  and let $B$ be embedded such that the boundary of the outer face $f$
  contains $u$, $v$, $s$, and $t$; see
  Figure~\ref{fig:corner-blocker}a.  In every Kandinsky
  representation, the total rotation around $f$ is~$-4$.  We show that
  this already implies that every Kandinsky representation has at
  least two bends.

  Let $f'$ be the inner face incident to $v$.  If $v$ has rotation $2$
  in $f'$ for a fixed Kandinsky representation, then one of the two
  edges $vs$ or $vt$ has rotation~$-1$ at $v$.  We assume without loss
  of generality that $vs$ has rotation~$-1$ at $v$, thus we get the
  following rotation values (see Figure~\ref{fig:corner-blocker}c):
  $\rot_{f'}(v) = 2$, $\rot_f(v) = -2$, $\rot_{f'}(vs[v]) = -1$, and
  $\rot_f(vs[v]) = 1$.  As $v$ has degree~2, we obtain another
  Kandinsky representation by setting $\rot_{f'}(v) = 1$, $\rot_{f}(v)
  = -1$, $\rot_{f'}(vs[v]) = 0$, and $\rot_f(vs[v]) = 0$; see
  Figure~\ref{fig:corner-blocker}d.  As this new Kandinsky
  representation has fewer bends, we can assume in the following that
  $\rot_{f'}(v) \not= 2$, which shows that the rotation at $v$ in $f$
  is at least~$-1$.  Clearly, the same holds for $u$.

  A similar argument shows that the rotations at $s$ and $t$ in $f$
  are at least~0.  It follows that the total rotation of vertices in
  the outer face is at least~$-2$.  Thus, to get a total rotation
  of~$-4$, there need to be two bends on edges incident to the outer
  face, which shows the claim.
\end{proof}

\begin{lemma}
  \label{lem:corner-blockers-no-corner-three-bends}
  Every Kandinsky representation of a corner blocker that blocks no
  corner of its attachment vertex has at least three bends.
\end{lemma}
\begin{proof}
  As shown in the proof of Lemma~\ref{lem:corner-blocker-lower-bound},
  one can reduce the number of bends of a Kandinsky representation of
  a corner blocker if the rotation at the attachment vertex in the
  outer face is~$-2$; see Figure~\ref{fig:corner-blocker}c,d.
\end{proof}

We can make the corner blockers stronger by nesting them.  The
\emph{nested corner blocker $B_d$ of depth $d$} is obtained by taking
$d$ corner blockers and identifying their attachment vertices.  The
nested corner blocker $B_d$ is embedded such that $v$ lies on the
outer face and the innermost face has distance $d$ to the outer face
(in the dual graph); see Figure~\ref{fig:corner-blocker-nested} for an
example.  Clearly, the statements from
Lemma~\ref{lem:corner-blocker-lower-bound} and
Lemma~\ref{lem:corner-blockers-no-corner-three-bends} extend to nested
corner blockers, where all bend numbers have to be multiplied with
$d$.

\begin{figure}
  \centering
  \includegraphics[page=2]{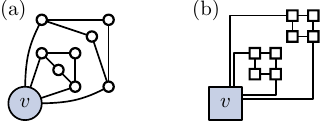}
  \caption{(a) The nested corner blocker $B_2$ of depth~$2$.
    (b)~Kandinsky representation of $B_2$ with~$4$ bends.}
  \label{fig:corner-blocker-nested}
\end{figure}

\paragraph{One-Bend Gadget}

Let $\Gamma$ be the graph consisting of the $2 \times 3$-grid with the
two columns $v_1,v_2,v_3$ and $u_1,u_2,u_3$ (from bottom to top)
together with the vertex $v$ connected to $v_1$, $v_2$, and $v_3$ and
the vertex $u$ connected to $u_2$; see
Figure~\ref{fig:one-bend-gadget}a.  We call $\Gamma$ the
\emph{one-bend gadget} with its two \emph{endvertices} $u$ and~$v$.
The path $\pi = (u, u_2, v_2, v)$ is called the \emph{bending path} of
$\Gamma$.  In the following we show that the bending path of a
one-bend gadget is (more or less) forced to have either rotation $1$
or $-1$ in every Kandinsky representation.  As for the corner blocker,
we say that the one-bend gadget blocks $k$ corners of the vertex $v$
in a given Kandinsky representation if $k$ corners of $v$ lie in the
inner face of $\Gamma$.

\begin{figure}
  \centering
  \includegraphics[page=1]{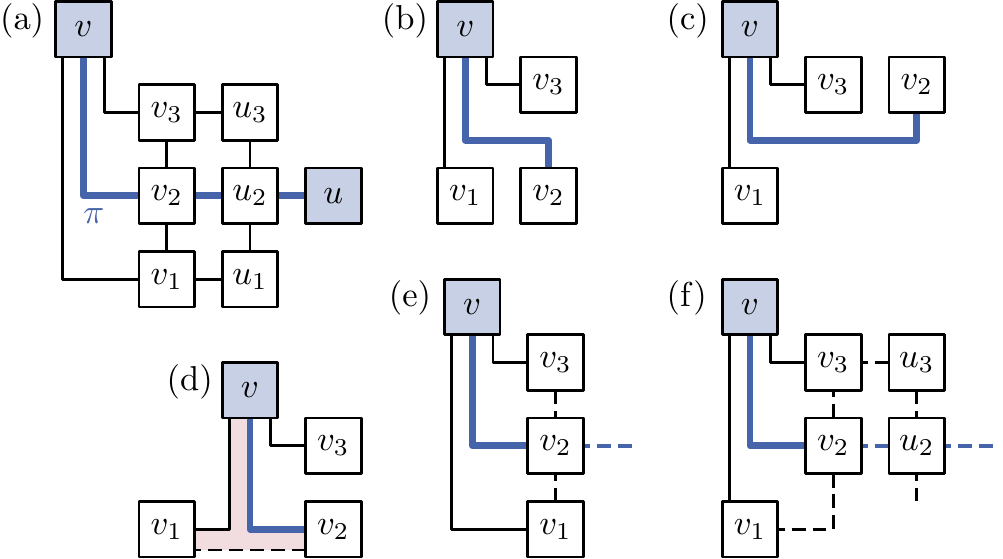}
  \caption{(a)~The one bend gadget $\Gamma$ with three bends without
    blocked corners at $v$ with $|\rot(\pi)| = 1$.  (b--c)~The two
    possible representations of the edges incident to $v$, when $vv_2$
    (blue) has two bends.  Both cannot be extended to $\Gamma$ without
    additional bends.  (d--f)~The possible ways to draw $vv_1$, when
    $vv_2$ (blue) has one bend together with the implied drawing on
    some other edges (dashed).  Either we get no Kandinsky
    representation of $\Gamma$ (d), or the path $\pi$ has absolute
    rotation~$1$ (e--f).}
  \label{fig:one-bend-gadget}
\end{figure}

\begin{lemma}
  \label{lem:one-bend-gadget-3-bends}
  Let $\mathcal K$ be a bend-minimal Kandinsky representation of the
  one-bend gadget $\Gamma$ blocking no corner of its degree-3
  endvertex.  Then $\mathcal K$ has three bends and the rotation of
  the bending path in $\Gamma$ is either~$1$ or~$-1$.
\end{lemma}
\begin{proof}
  Note that the Kandinsky representation of $\Gamma$ in
  Figure~\ref{fig:one-bend-gadget}a does not block a corner of $v$
  and has three bends.  It remains to show that this drawing is
  optimal and that the rotation of the bending path $\pi$ is
  always~$1$ or~$-1$.

  We consider all Kandinsky representations of $\Gamma$ with at most
  three bends blocking no corner of $v$ and show that each of these
  representations has three bends and rotation~$1$ or~$-1$ on $\pi$.
  We start with two simple facts.  First, blocking no corner of $v$
  requires that at least two of the edges $vv_1$, $vv_2$, and $vv_3$
  to have a bend, as they all leave $v$ at the same side.  Second, the
  edges in each of the triangles $vv_1v_2$ and $vv_2v_3$ require at
  least two bends, as they have rotation~$2$ at $v$.  We use these
  facts several times in the following case distinction on the number
  of bends of $vv_2$.

  Assume \textbf{$vv_2$ has three bends}.  As at least two of the
  edges $vv_1$, $vv_2$, and $vv_3$ have bends, we get at least four
  bends in total (but we consider only drawings with at most three
  bends).  Assume that \textbf{$vv_2$ has zero bends}.  Then the two
  bends of the triangle $vv_1v_2$ must be on the edges $vv_1$ and
  $v_1v_2$ and the bends of the triangle $vv_2v_3$ must be on the
  edges $vv_3$ and $v_2v_3$.  Thus, there are at least four bends,
  which again contradicts the restriction to at most three bends.

  If \textbf{$vv_2$ has two bends}, one of the two edges $vv_1$ or
  $vv_3$ has one bend, the other has no bend (as we have more than
  three bends otherwise).  Assume that $vv_3$ has one bends, the other
  case is symmetric.  As $vv_1$ has no bend, the direction of the bend
  of $vv_3$ and of the first bend of $vv_2$ is fixed.  The remaining
  choice is the second bend of $vv_2$; see
  Figure~\ref{fig:one-bend-gadget}b and~c for an illustration of
  the two possible Kandinsky representations.  Since we already used
  three bends, the $2\times3$ grid consisting of the nodes $v_1 \dots
  v_3$ and $u_1 \dots u_3$ must be drawn without any bends.  However,
  the Kandinsky representation without bends of the $2\times3$ grid is
  unique (see Figure~\ref{fig:one-bend-gadget}a) and can obviously
  not be merged with one of the Kandinsky representations of the three
  edges incident to $v$ shown in Figure~\ref{fig:one-bend-gadget}b
  and~c.

  Assume that \textbf{$vv_2$ has one bend}.  Assume without loss of
  generality that the bend on $vv_2$ is a left bend, when traversing
  it from $v$ to $v_2$ (i.e., $vv_2$ has rotation~$1$ in the triangle
  $vv_2v_3$).  This implies that the edge $vv_3$ has at least one bend
  as in Figure~\ref{fig:one-bend-gadget}d--f.  If $vv_1$ has a bend
  but in the other direction then all remaining edges have to be
  straight, which is not possible for $v_1v_2$ without creating an
  empty triangle (Figure~\ref{fig:one-bend-gadget}d).  Thus,
  $v_1v_2$ has either a bend in the same direction as $vv_1$ or no
  bend.  Consider the former case first; see
  Figure~\ref{fig:one-bend-gadget}e.  We split the bending path
  $\pi$ into two parts, the edge $vv_2$ and the path from $v_2$ to
  $u$.  Clearly, the absolute rotation of $vv_2$ is~$1$.  As we
  already used three bends, the Kandinsky representation of $\Gamma -
  v$ is unique.  Thus, the path from $v_2$ to $u$ must have
  rotation~$0$.  To show $|\rot(\pi)| = 1$, it remains to show that
  the rotation of $\pi$ at $v_2$ is~$0$, which is the case if there is
  no~$0^\circ$ angle at~$v_2$.  This angle would have to be adjacent
  to the edge~$vv_2$ as it is the only one having a bend.  However,
  $vv_2$ has only one bend and, since~$\rot_f(vv_2[v]) = 1$, it
  follows that~$\rot_f(vv_2[v_2]) = 0$, where~$f$ is the face bounded
  by~$vv_2v_3$.  But then there can be no~$0^\circ$ bend at~$v_2$.

  It remains to deal with the case that $vv_1$ has no bend; see
  Figure~\ref{fig:one-bend-gadget}f.  As the triangle $vv_1v_2$
  needs two bends, the edge $v_1v_2$ must be drawn with a bend.  All
  remaining edges must have zero bends, as three bends are already
  used.  As before, this shows that the subpath of $\pi$ from $v_2$ to
  $u$ has rotation~$0$ and for $|\rot(\pi)| = 1$ it remains to show
  that the rotation of $\pi$ at $v_2$ is~$0$.  To this end, consider
  the triangle $vv_2v_3$ and the quadrangle $v_2u_2u_3v_3$.  In the
  quadrangle, all edges are straight lines, which ensures that the
  rotation at $v_2$ is~$1$.  In the triangle, the rotation at $v_2$
  must also be~$1$ (otherwise the rotation at $v_3$ would need to
  be~$2$, but there is no edge that can assign its bend to this
  $0^\circ$ angle).  Thus, the rotation of $v_2$ in the path $\pi$
  is~$0$, which shows $|\rot(\pi)| = 1$.

  It follows that the path $\pi$ has rotation $1$ or $-1$ in every
  Kandinsky representation of $\Gamma$ that has only three bends and
  blocks no corner of $v$.  Moreover, we showed that all such
  Kandinsky representations require three bends.
\end{proof}

\subsubsection*{Putting Things Together}

Let $G = (V, E = E_0\cupdot E_1)$ (together with a combinatorial
embedding) be an instance of \textsc{Orthogonal 01-Embeddability}.  We
assume that the angles at vertices that are incident to a~$1$-edge are
fixed.  We construct an embedded graph $G'$ that then serves as
instance of \textsc{Kandinsky Bend Minimization}.  To construct $G'$,
we start with~$G$.  Let $v$ be a vertex incident to the face $f$.  If
the angle of~$v$ in~$f$ is fixed to~$\alpha$, we attach
$\alpha/90^\circ$ nested corner blockers of depth $4$ to $v$ embedded
next to each other into the face $f$; see
Figure~\ref{fig:minimization-hardness}a.  Otherwise, if the angle is
not fixed, we attach a single corner blocker of depth~$4$ at~$v$
in~$f$.  By suitably increasing the depth of some corner blockers we
ensure that each vertex is incident to exactly 16 corner blockers; see
Figure~\ref{fig:minimization-hardness}b.  This is not strictly
necessary but simplifies some of our computations.  Finally, we
replace every edge $uv \in E_1$ (i.e., every edge that requires one
bend) by a copy of the one-bend gadget $\Gamma$, identifying $u$ and
$v$ with the endvertices of $\Gamma$.  Note that, by assumption,
both~$u$ and~$v$ have four corner blockers of depth~4.  To obtain the
following theorem, we show that the resulting graph $G'$ admits a
Kandinsky representation with at most $32|V| + 3|E_1|$ bends if and
only if $G$ admits an orthogonal 01-embedding (note that deciding
whether a planar embedded graph admits a Kandinsky representation with
at most $k$ bends is clearly in NP).

\begin{figure}
  \centering
  \includegraphics[page=1]{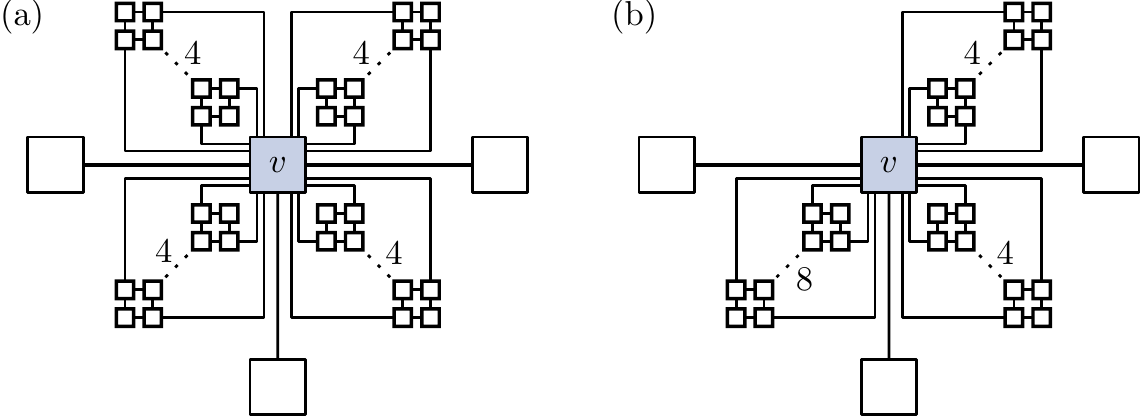}
  \caption{(a--b)~A degree-3 vertex with and without fixed angles and
    attached corner blockers, respectively.}
  \label{fig:minimization-hardness}
\end{figure}

\begin{theorem}
  \textsc{Kandinsky Bend Minimization} is NP-complete.
\end{theorem}
\begin{proof}
  Let $G$ be an instance of \textsc{Orthogonal 01-Embeddability} (with
  fixed angles at vertices incident to~$1$-edges) and let $G'$ be the
  corresponding instance of \textsc{Kandinsky Bend Minimization}.
  Assume we have an orthogonal 01-representation $\mathcal O$ of $G$.
  We show how to construct an Kandinsky representation $\mathcal K$ of
  $G'$ with $32|V| + 3|E_1|$ bends.  We interpret $\mathcal O$ as a
  Kandinsky representation.  We first add the nested corner blockers
  to the representation.  Let $v$ be a vertex with incident face $f$.
  By construction, $v$ has at least as many corners in $f$ as there
  are nested corner blockers incident to $v$ embedded in the face $f$.
  Thus, these corner blockers can be added with~2 bends for each
  corner blocker (see the drawing in
  Figure~\ref{fig:corner-blocker}b).  This yields $32|V|$ bends in
  total.

  Moreover, the drawings of the 1-edges can be replaced by drawings of
  one-bend gadgets with three bends (the drawing in
  Figure~\ref{fig:one-bend-gadget}a or the symmetric drawing where
  the bending path $\pi$ is bent to the other direction).  This yields
  $3|E_1|$ bends for all 1-edges.  Hence, we get a Kandinsky
  representation $\mathcal K$ of $G'$ with $32|V| + 3|E_1|$ bends in
  total.

  For the opposite direction, we show that a Kandinsky representation
  $\mathcal K$ of $G'$ with at most $32|V| + 3|E_1|$ bends implies the
  existence of a orthogonal 01-embedding $\mathcal O$ of $G$.  We show
  that the following three facts hold for $\mathcal K$.
  \begin{enumerate}
  \item Every nested corner blocker blocks a corner.
  \item Every one-bend gadget has three bends and blocks no corner of
    its degree-3 endvertex in $\mathcal K$.
  \item All remaining edges (the edges in $E_0$) have~0 bends.
  \end{enumerate}
  We use a charging argument assigning the costs for bends either to
  corner blockers, to one-bend gadgets or to the edges in $E_0$, such
  that the total cost is at most the total number of bends.  By
  Lemma~\ref{lem:corner-blocker-lower-bound}, every corner blocker of
  depth~$d$ requires at least $2d$ bends.  Moreover, if such a nested
  corner blocker does not block a corner, it has at least $3d$ bends
  (Lemma~\ref{lem:corner-blockers-no-corner-three-bends}).  For corner
  blockers that block a corner, we charge cost~$2d$ (which is equal to
  the number of bends).  For corner blockers blocking no corner, we
  charge cost~$2d+1$ (which is~$d-1 \ge 3$ less than the number of
  bends; note that all corner blockers have depth at least~$4$).  A
  one-bend gadget with more than three bends is charged cost~$4$.  A
  one-bend gadget that does not block a corner of its degree-3
  endvertex has at least three bends by
  Lemma~\ref{lem:one-bend-gadget-3-bends} and we charge cost~$3$ for
  it.  If a one-bend gadget blocks a corner of its degree-3 endvertex,
  then at least one of the adjacent nested corner blockers does not
  block a corner.  As we charged cost~$2d+1$ for this corner blocker
  although it has at least $3d$ bends, we can again charge
  cost~$3d-(2d+1) = d-1 \ge 3$ for the one-bend gadget.  For the
  remaining edges in $E_0$ we simply charge cost equal to the number
  of bends.

  Hence, every nested corner blocker of depth~$d$ is charged at least
  cost~$2d$ and every one-bend gadget is charged at least cost~$3$.
  Recall that there are $16|V|$ corner blockers and $|E_1|$ one bend
  gadgets.  To get a total cost of at most $32|V| + 3|E_1|$, every
  corner blocker of depth~$d$ must be charged exactly cost~$2d$, which
  implies that it blocks a corner and thus shows the first fact.
  Since the endvertices of one-bend gadgets are incident to four
  corner blockers, each of which indeed blocks a corner, this also
  implies that no one-bend gadget can block a corner of its degree-3
  endvertex.  Thus, by Lemma~\ref{lem:one-bend-gadget-3-bends}, every
  one-bend gadget has at least three bends.  Moreover, every one-bend
  gadget has no more than three bends as it would otherwise be charged
  cost~$4$, which shows the second fact.  The third fact follows as
  the cost charged to edges in $E_0$ must be~$0$.

  By the second fact and Lemma~\ref{lem:one-bend-gadget-3-bends} the
  bending path of every one-bend gadget has absolute rotation~$1$ in
  $\mathcal K$.  Thus, we can replace each one-bend gadget by an edge
  with exactly one bend. Removing the corner blockers yields a
  representation of $G$ in which no two edges leave a common incident
  vertex on the same side, as every nested corner blocker blocks a
  corner (first fact).  Moreover, the edges in $E_0$ have zero bends.
  Hence, the resulting representation of $G$ is an orthogonal
  representation (and not only a Kandinsky representation) and the
  edges in $E_1$ and $E_0$ have one and zero bends, respectively.
\end{proof}

\begin{theorem}
  \textsc{Kandinsky Bend Minimization} is NP-complete, even if we
  allow empty faces or require every edge to have at most one bend (or
  both).
\end{theorem}
\begin{proof}
  That the problem remains NP-hard when we require each edge to have
  at most one bend is obvious, as all Kandinsky representations
  involved in the construction above have at most one bend per edge.
  In fact, this requirement would even make some arguments simpler.
  The only place where we argued with empty faces is in the proof of
  Lemma~\ref{lem:one-bend-gadget-3-bends} to exclude the situation
  shown in Figure~\ref{fig:one-bend-gadget}d.  It is not hard to see
  that this situation can also be excluded when allowing empty faces,
  as even in this case, it is not possible to complete the embedding
  without additional bends.
\end{proof}

\section{A Subexponential Algorithm}
\label{sec:subexp-algor}

In this section, we give an algorithm for computing optimal
Kandinsky representations of planar graphs with fixed planar embedding
in subexponential running time.  To this end, we use dynamic
programming on sphere cut decompositions, which are special types of
branch decompositions~\cite{dpbf-eeapg-10}.

The basic idea is as follows.  Consider two graphs $G_1$ and $G_2$
with disjoint edge sets that share a set of \emph{attachment
  vertices}.  We assume that the union $G$ of $G_1$ and $G_2$ is
planar and has a fixed planar embedding.  We say that $G_1$ and $G_2$
are \emph{glueable} if both graphs are connected and there is a simple
closed curve in the embedding of $G$ that separates $G_1$ from $G_2$
(note that this curve must contain the attachment vertices); see
Figure~\ref{fig:equivalent-representations-glueable}.  We also say
that $G_1$ ($G_2$) is a \emph{glueable subgraph} of $G$.

\begin{figure}
  \centering
  \includegraphics[page=1]{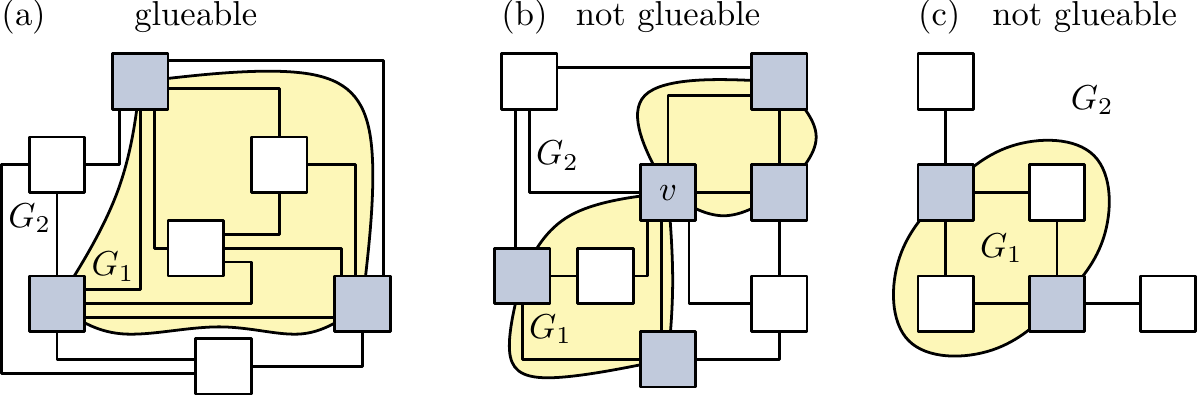}
  \caption{(a)~The decomposition of a graph into two glueable
    subgraphs $G_1$ and $G_2$.  The attachment vertices are shaded
    blue.  (b)~This decomposition is not glueable, as a closed curve
    separating $G_1$ from $G_2$ cannot be simple ($v$ must be visited
    twice).  (c)~The decomposition is not glueable since $G_2$ is
    disconnected.}
  \label{fig:equivalent-representations-glueable}
\end{figure}

Now assume we know two Kandinsky representations $\mathcal K_1$ and
$\mathcal K_2$ of $G_1$ and $G_2$.  Depending on $\mathcal K_1$ and
$\mathcal K_2$ one might be able to merge them into a Kandinsky
representation of the whole graph $G$.  We can generate every
Kandinsky representation of $G$ in this way, by merging every
representation of $G_1$ with every representation of $G_2$.  Clearly,
considering all pairs of representations of $G_1$ and $G_2$ is not
efficient.  Thus, we group Kandinsky representations of $G_1$ that
behave the same with respect to merging them with representations of
$G_2$ into equivalence classes.  If we know an optimal Kandinsky
representation for each equivalence class of $G_1$ and $G_2$, it is
sufficient to merge those optimal representatives of equivalence
classes to obtain an optimal representation of $G$.  If $G$ is
hierarchically decomposed, one can start with optimal Kandinsky
representations of the edges and merge them step by step to obtain
$G$.

In the following we first characterize which Kandinsky representations
of a glueable subgraph are equivalent in the sense that they can be
merged with the same Kandinsky representation of the remaining graph
(Section~\ref{sec:interf-kand-repr}).  Afterwards, we estimate in how
many different ways the Kandinsky representations of subgraphs can be
merged into one (Section~\ref{sec:merg-two-kand}).  Finally, we
conclude with the algorithm and some interesting special cases
(Section~\ref{sec:algorithm}).

\subsection{Interfaces of Kandinsky Representations}
\label{sec:interf-kand-repr}

Let $\mathcal K$ be a Kandinsky representation of $G$ and let
$\mathcal K_1$ be the representation induced on $G_1$.  Let $\mathcal
K_1'$ be another Kandinsky representation of $G_1$.  By
\emph{replacing $\mathcal K_1$ with $\mathcal K_1'$ in $\mathcal K$}
we mean the following.  Every rotation value in $\mathcal K$ involving
only edges belonging to $G_1$ are set to the value specified in
$\mathcal K_1'$ while all other values remain as they are.  In other
words the following rotations in $\mathcal K$ are changed to their
value in $\mathcal K_1'$: $\rot(e)$ if the edge $e$ belongs to $G_1$;
$\rot(uv[u])$ and $\rot(uv[v])$ (the rotation of $uv$ at the vertices
$u$ and $v$) if $uv$ belongs to $G_1$; and $\rot(e_1, e_2)$ (for two
edges $e_1$ and $e_2$ incident to a common vertex) if both edges $e_1$
and $e_2$ belong to $G_1$.  Note that the resulting set of rotation
values is not necessarily a Kandinsky representation, as some
properties of Kandinsky representations might be violated.

We say that the two Kandinsky representations $\mathcal K_1$ and
$\mathcal K_1'$ of $G_1$ have \emph{the same interface} if replacing
$\mathcal K_1$ with $\mathcal K_1'$ (and vice versa) in any Kandinsky
representation of $G$ yields a Kandinsky representation of $G$.  We
will see later (Lemma~\ref{lem:kandinksy-replacement}) that it does
not depend on the remaining graph $G\setminus G_1$, whether two
representations of $G_1$ have the same interface.  The two Kandinsky
representations in Figure~\ref{fig:equivalent-representations}a have
the same interface.  Clearly, having the same interface is an
equivalence relation.  We call the equivalence classes of this
relation the \emph{interface classes}.

\begin{figure}
  \centering
  \includegraphics[page=2]{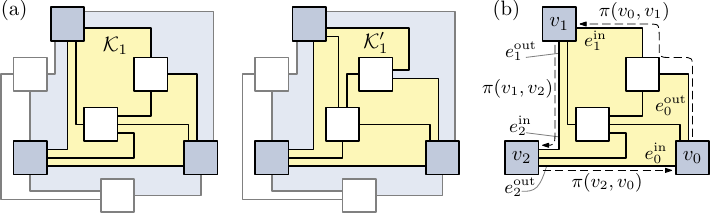}
  \caption{(a)~The graph $G$ with glueable subgraph $G_1$ (shaded
    yellow).  The attachment vertices and the faces shared by $G_1$
    and $G_2$ are blue.  The two Kandinsky representations $\mathcal
    K_1$ (left) and $\mathcal K_1'$ (right) of $G_1$ are
    interchangeable: In any representation of $G$ inducing $\mathcal
    K_1$ on $G_1$, one can replace $\mathcal K_1$ by $\mathcal K_1'$
    and vice versa.  (b)~Illustration of the notation used to define
    the interface paths, the attachment rotations and the $0^\circ$
    flags.}
  \label{fig:equivalent-representations}
\end{figure}

Now consider again two glueable subgraphs $G_1$ and $G_2$ of a plane
graph $G$.  Since $G_1$ and $G_2$ are glueable, we know that $G_2$
lies in a single face $f$ of $G_1$.  Let $C_f$ be the facial cycle of
$f$ and assume for now, that $C_f$ is simple (i.e., $G_1$ contains no
cutvertex incident to $f$).  Let $v_0, \dots, v_\ell$ be the
attachment vertices appearing in that order in $C_f$ (clockwise for
inner, counter-clockwise for outer faces).  This decomposes $C_f$ into
the paths $\pi_f(v_0, v_1), \pi_f(v_1, v_2), \dots, \pi_f(v_\ell,
v_0)$, which we call \emph{interface paths}.  As the face we consider
is unique, we often omit the subscript and simply write $\pi(v_i,
v_{i+1})$.  Moreover, the values for $i$, $i-1$ and $i+1$ are always
meant modulo $\ell+1$.  For an attachment vertex $v_i$, denote the
last edge of the path $\pi(v_{i-1},v_i)$ by $e_i^\inedge$ and the
first edge of the path $\pi(v_i, v_{i+1})$ by $e_i^\outedge$; see
Figure~\ref{fig:equivalent-representations}b.

Let $\mathcal K_1$ and $\mathcal K_1'$ be two Kandinsky
representations of $G_1$.  We say that $\mathcal K_1$ and $\mathcal
K_1'$ have \emph{compatible interface paths} if $\pi(v_i, v_{i+1})$
has the same rotation in $\mathcal K_1$ and $\mathcal K_1'$ (i.e.,
$\rot_{\mathcal K_1}(\pi(v_i, v_{i+1})) = \rot_{\mathcal
  K_1'}(\pi(v_i, v_{i+1}))$) for every $i = 1, \dots, k$.  Moreover,
$\mathcal K_1$ and $\mathcal K_1'$ have \emph{the same attachment
  rotations} if for every attachment vertex $v_i$, the rotation
$\rot(e_i^\inedge, e_i^\outedge)$ is the same in $\mathcal K_1$ and
$\mathcal K_1'$.  In Figure~\ref{fig:equivalent-representations}b,
the interface paths $\pi(v_0, v_1)$, $\pi(v_1, v_2)$, and $\pi(v_2,
v_0)$ have rotations~$-1$,~$1$, and~$0$, respectively, and the
attachment rotations at the vertices $v_0$, $v_1$, and $v_2$
are~$-1$,~$-1$, and~$-2$, respectively.

When considering orthogonal representations (with maximum degree~4)
and not Kandinsky representations, having compatible interface paths
and the same attachment rotations is sufficient for two
representations to have the same interface.  In case of Kandinsky
representations, we have to care about $0^\circ$ angles at the
attachment vertices.  Thus, for an attachment vertex $v_i$, the
rotations at the end $v_i$ of the edges $e_i^\inedge$ and
$e_i^\outedge$ ($\rot(e_i^\inedge[v_i])$ and
$\rot(e_i^\outedge[v_i])$), which can take the values $-1$, $0$, or
$1$ are of importance.  The actual value of these rotations is not
important, we only care about whether they are~$-1$ or something else.
We call these information the \emph{$0^\circ$ flags}, which has the
value \texttt{true} for a rotation of $-1$ and \texttt{false}
otherwise.  We say that $\mathcal K_1$ and $\mathcal K_1'$ have the
\emph{same $0^\circ$ flags} if all their $0^\circ$ flags have the same
values.  Possible values for the $0^\circ$ flags in
Figure~\ref{fig:equivalent-representations}b are \texttt{true} for
$e_0^\outedge[v_0]$ and for $e_1^\inedge[v_1]$ and \texttt{false} for
all other flags.

In case the facial cycle $C_f$ is not simple, it might contain an
attachment vertex $v_i$ several times.  However, since $G_1$ and $G_2$
are glueable, the simple closed curve separating $G_1$ from $G_2$
gives an order of the attachment vertices.  We simply take this order
to define the interface paths.  All remaining definitions work as
before.  

\begin{lemma}
  \label{lem:kandinksy-replacement}
  Two Kandinsky representations have the same interface if and only if
  they have compatible interface paths, the same attachment rotations,
  and the same $0^\circ$ flags.
\end{lemma}
\begin{proof}
  We first show the only-if part.  Let $G$ be a plane graph with
  Kandinsky representation $\mathcal K$ with restriction $\mathcal
  K_1$ to the glueable subgraph $G_1$.  Let $\mathcal K_1'$ be another
  Kandinsky representation of $G_1$.  Assume there is an interface
  path $\pi$ that has a different rotation in $\mathcal K_1$ than in
  $\mathcal K_1'$.  Let $f$ be the face incident to $\pi$ shared by
  $G_1$ and the remaining graph $G_2$ (one of the blue faces in
  Figure~\ref{fig:equivalent-representations}a).  By replacing
  $\mathcal K_1$ with $\mathcal K_1'$ the rotation of $\pi$ in $f$
  changes, but all other rotations in $f$ stay the same.  Thus, the
  total rotation around $f$ cannot be~$4$ ($-4$ if $f$ is the outer
  face), which shows that the resulting set of rotations is not a
  Kandinsky representation of $G$ (contradiction to
  Property~(\ref{item:rotation-of-face})).  Hence, $\mathcal K_1$ and
  $\mathcal K_1'$ do not have the same interface.  A similar argument
  shows that having the same attachment rotations is necessary, since
  otherwise the total rotation around a vertex would change by
  replacing $\mathcal K_1$ with $\mathcal K_1'$, which contradicts
  Property~(\ref{item:rotation-around-vertex}).

  Finally, assume $\mathcal K_1$ and $ \mathcal K_1'$ have different
  $0^\circ$ flags.  Thus, there exists an attachment vertex $v$ with
  incident edge $e_1$ (belonging to an interface path) such that
  $\rot(e_1[v])$ is (without loss of generality) $-1$ in $\mathcal
  K_1$ (value \texttt{true}) and $0$ or $1$ in $\mathcal K_1'$ (value
  \texttt{false}).  As $v$ is an attachment vertex, the remaining
  graph $G_2$ contains an edge incident to $v$.  Let $e_2$ be the edge
  of $G_2$ incident to $v$ that shares a face $f$ with $e_1$.  Then
  one might choose the Kandinsky representation $\mathcal K$ of $G$
  such that the rotation $\rot_f(e_1, e_2)$ at $v$ in $f$ is~$2$
  (angle of $0^\circ$) while $\rot_f(e_2[v])$ is $0$ or $1$.  Then
  $\rot(e_1[v])$ must be $-1$ by
  Property~(\ref{item:0-deg-bend-assignment}), which is true for
  $\mathcal K_1$ but not for $\mathcal K_1'$.  Hence, replacing
  $\mathcal K_1$ with $\mathcal K_1'$ does not yield a Kandinsky
  representation of $G$, which shows that having the same $0^\circ$
  flags is also necessary for having the same interface.

  For the other direction, let $G_1$ and $G_2$ be glueable graphs with
  union $G$ and let $\mathcal K$ be a Kandinsky representation of $G$
  with restrictions $\mathcal K_1$ and $\mathcal K_2$ to $G_1$ and
  $G_2$, respectively.  Let $\mathcal K_1'$ be a Kandinsky
  representation of $G_1$ with compatible interface paths, the same
  attachment rotations, and the same $0^\circ$ flags.  We show that
  replacing $\mathcal K_1$ with $\mathcal K_1'$ in $\mathcal K$ yields
  a Kandinsky representation of $G$ by showing that the resulting
  rotation values satisfy properties
  (\ref{item:rotation-of-face})--(\ref{item:0-deg-bend-assignment})
  from Section~\ref{sec:kand-repr}.

  Property~(\ref{item:consistent-edges}) is trivially satisfied, as
  all rotations concerning a single edge come either from $\mathcal
  K_1'$ or from $\mathcal K_2$, which are both Kandinsky
  representations and thus satisfy this property.
  Property~(\ref{item:rotation-range-at-vertex}) is also satisfied, as
  the rotation at a vertex either stays as it is in $\mathcal K$ or it
  is changed to its value in $\mathcal K_1'$ and thus lies in the
  interval $[-2, 2]$.

  For Property~(\ref{item:rotation-of-face}), consider a face $f$ of
  $G$.  If all edges in the boundary of $f$ belong to only one of the
  graphs $G_1$ and $G_2$, then the total rotation in $f$ is equal to
  its total rotation in $\mathcal K_1'$ or $\mathcal K_2$,
  respectively.  As $\mathcal K_1'$ and $\mathcal K_2$ are Kandinsky
  representations, they satisfy
  Property~(\ref{item:rotation-of-face}).  If the boundary of $f$
  contains edges from both graphs $G_1$ and $G_2$ (one of the blue
  faces in Figure~\ref{fig:equivalent-representations}a), it is
  composed of two interface paths $\pi_1$ and $\pi_2$ belonging to
  $G_1$ and $G_2$, respectively, that share their endvertices $u$ and
  $v$.  By replacing $\mathcal K_1$ with $\mathcal K_1'$, the
  representation of $\pi_2$ does not change.  Moreover, the rotations
  at $u$ and $v$ in $f$ remain unchanged.  The representation of
  $\pi_1$ might of course change, however, the rotation remains the
  same as $\mathcal K_1$ and $\mathcal K_1'$ have compatible interface
  paths.

  A similar argument shows that
  Property~(\ref{item:0-deg-bend-assignment}) is satisfied.  Let $v$
  be a vertex with rotation~$2$ (corresponding to an angle of
  $0^\circ$) in a face $f$, i.e., $\rot_f(uv,vw) = 2$.  We only need
  to consider the case where (without loss of generality) $uv$ belongs
  to $G_1$ and $vw$ belongs to $G_2$, as all other cases are trivial.
  Then $\rot_f(uv[v]) = -1$ or $\rot_f(vw[v]) = -1$ holds in $\mathcal
  K$.  In the latter case, $\rot_f(vw[v])$ does not change by
  replacing $\mathcal K_1$ with $\mathcal K_1'$ as $vw$ belongs to
  $G_2$.  In the former case, $\rot_f(uv[v]) = -1$ implies that the
  corresponding $0^\circ$ flag in $\mathcal K_1$ is \texttt{true}.  As
  $\mathcal K_1$ and $\mathcal K_1'$ have the same $0^\circ$ flags,
  this flag is also \texttt{true} in $\mathcal K_1'$, which implies
  that $\rot_f(uv[v]) = -1$ is still true in $\mathcal K_1'$ and thus
  in $\mathcal K'$.

  Finally, to show Property~(\ref{item:rotation-around-vertex}),
  consider a vertex $v$.  If $v$ is not an attachment vertex, all
  rotations at $v$ come either from $\mathcal K_1'$ or from $\mathcal
  K_2$ and thus satisfy Property~(\ref{item:rotation-around-vertex}).
  Let $v$ be an attachment vertex and let $f_1$ be the face of $G_1$
  that completely contains $G_2$.  The only rotations at $v$ that
  might change by replacing $\mathcal K_1$ with $\mathcal K_1'$ are
  the rotations in faces not shared with $G_2$.  These are exactly the
  faces of $G_1$ incident to $v$ except for $f_1$.  As $\mathcal K_1$
  and $\mathcal K_1'$ have the same rotations at attachment vertices,
  the rotation $\rot_{f_1}(v)$ is the same in $\mathcal K_1$ and
  $\mathcal K_1'$.  Thus, by
  Property~(\ref{item:rotation-around-vertex}) the sum of all other
  rotations around $v$ in $G_1$ must also be the same in both
  representations $\mathcal K_1$ and $\mathcal K_1'$.  Hence, the
  total sum of rotations at $v$ does not change by replacing $\mathcal
  K_1$ with $\mathcal K_1'$, which concludes the proof.
\end{proof}

It follows that each interface class is uniquely described by the
rotations of the interface paths, by the rotations at the attachment
vertices, and by the values of the $0^\circ$ flags.  We simply call
this set of information the \emph{interface} of $G_1$ ($G_2$) in $G$.
Note that this redefines what it means for two Kandinsky
representations to have the same interface.  However, the definitions
are consistent due to Lemma~\ref{lem:kandinksy-replacement} and we
will use them interchangeably.

\subsection{Merging two Kandinsky Representations}
\label{sec:merg-two-kand}

So far, we considered the case that there is a Kandinsky
representation $\mathcal K$ of $G$ that can be altered by replacing
the Kandinsky representation of the subgraph $G_1$.  Now we change the
point of view and assume we have Kandinsky representations $\mathcal
K_1$ and $\mathcal K_2$ of $G_1$ and $G_2$, respectively, that we want
to combine to get a Kandinsky representation of $G$.  We say that
$\mathcal K_1$ and $\mathcal K_2$ can be \emph{merged} if there exists
a Kandinsky representation $\mathcal K$ of $G$ whose restriction to
$G_1$ and $G_2$ is $\mathcal K_1$ and $\mathcal K_2$, respectively.
Note that the only rotations in $\mathcal K$ that occur neither in
$\mathcal K_1$ nor in $\mathcal K_2$ are rotations at attachment
vertices between an edge of $G_1$ and an edge of $G_2$.  We call these
rotations the \emph{shared rotations}; see
Figure~\ref{fig:merging-step}a.  Thus, merging $\mathcal K_1$ and
$\mathcal K_2$ is the process of choosing values for the shared
rotation, such that the resulting set of rotations is a Kandinsky
representation of $G$.

In the following, we consider the case where $G$ itself is a glueable
subgraph of a larger graph $H$.  We call this the \emph{merging step}
$G = G_1 \sqcup G_2$.  Note that $G_1$ and $G_2$ are not only glueable
subgraphs of $G$ but also of $H$.  Note further that the interface of
$G_1$ ($G_2$) in $G$ can be deduced from the interface of $G_1$
($G_2$) in~$H$.  When dealing with a merging step, we always consider
the interfaces of $G_1$ and $G_2$ in $H$ (which contain more
information than their interfaces in $G$).  The \emph{width} of a
merging step is the maximum number of attachment vertices of $G_1$,
$G_2$, and $G$ in $H$; see Figure~\ref{fig:merging-step}b for an
example.

\begin{figure}
  \centering
  \includegraphics[page=1]{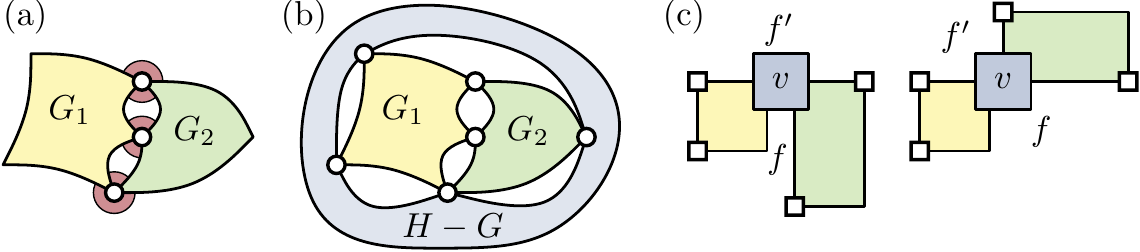}
  \caption{(a)~Merging $G_1$ and $G_2$.  The shared rotations are
    marked red.  (b)~Illustration of a merging step.  The width of
    this merging step is~$5$ ($G_1$ has $5$ attachment vertices).
    (c)~Two ways to choose the shared rotations.}
  \label{fig:merging-step}
\end{figure}

If the Kandinsky representations $\mathcal K_1$ and $\mathcal K_2$ can
be merged, then every Kandinsky representation $\mathcal K_1'$ with
the same interface as $\mathcal K_1$ can be merged in the same way
(i.e., with the same shared rotations) with $\mathcal K_2$ as one can
first merge $\mathcal K_1$ with $\mathcal K_2$ and then replace
$\mathcal K_1$ by $\mathcal K_1'$.  Moreover, the resulting Kandinsky
representations $\mathcal K$ and $\mathcal K'$ of $G$ have the same
interface for the following reason.  In every Kandinsky representation
of $H$ the representation $\mathcal K$ can be replaced by $\mathcal
K'$ as this is equivalent to replacing $\mathcal K_1$ by $\mathcal
K_1'$ (which can be done as $\mathcal K_1$ and $\mathcal K_1'$ have
the same interface).  Thus, the only choices that matter when merging
two Kandinsky representations are to choose shared rotations and
interfaces for $G_1$ and $G_2$.  Thus, the term of merging Kandinsky
representations extends to \emph{merging interfaces}.  We call a
choice of shared rotations and interfaces for $G_1$ and $G_2$
\emph{compatible}, if these interfaces can be merged using the chosen
rotations.

The following lemma bounds the number of compatible combinations.  It
is parametrized with the width $k$ of the merging step and the
\emph{maximum rotation}~$\rho$.  The maximum rotation of a graph $H$
is $\rho$ if $H$ admits an optimal Kandinsky representation such that
the absolute rotations of the interface paths in every glueable
subgraph of $H$ are at most $\rho$.  With the maximum rotation of a
merging step, we mean the maximum rotation of the whole graph~$H$.  We
give bounds for $\rho$ in Lemma~\ref{lem:rotation-interface-paths}.

\begin{lemma}
  \label{lem:merging-step-compatible-choices}
  In a merging step $G = G_1 \sqcup G_2$ of width $k$ with maximum
  rotation $\rho$, there are at most $(2\rho + 1)^{\left\lfloor
      1.5k\right\rfloor - 1}\cdot 330^k$ compatible choices for the
  shared rotations and the interfaces of $G_1$ and~$G_2$.
\end{lemma}
\begin{proof}
  Let $k_{12}$ be the number of attachment vertices shared by $G_1$
  and $G_2$ and let $k_1$ and $k_2$ be the number of exclusive
  attachment vertices of $G_1$ and $G_2$, respectively.  In the
  example in Figure~\ref{fig:merging-step}b, $k = 5$, $k_{12} = 3$,
  $k_1 = 2$, and $k_2 = 1$.  As $G_1$ and $G_2$ both have at most $k$
  attachment vertices, we have $k_1 + k_{12} \le k$ and $k_2 + k_{12}
  \le k$.  Moreover, every exclusive attachment vertex is an
  attachment vertex of $G$, thus $k_1 + k_2 \le k$ holds.  By summing
  these three inequalities we directly get $k_1 + k_2 + k_{12} \le
  \left\lfloor 1.5k\right\rfloor$.  We start with a rough estimation
  of the possible combinations and then show how to reduce the number
  by ruling out choices that are not compatible and thus will never
  lead to a Kandinsky representation.

  The graph $G_1$ has $k_1 + k_{12} \le k$ attachment vertices and
  thus also $k_1 + k_{12} \le k$ interface paths.  The absolute
  rotation of each interface path is at most $\rho$, thus there are at
  most $2\rho + 1$ possible values for those rotations.  This leads to
  at most $(2\rho + 1)^k$ combinations.  For every attachment vertex
  there is the attachment rotation that can be any of the five
  integers in $[-2, 2]$.  Moreover, there are two binary $0^\circ$
  flags for each attachment vertex which gives $5\cdot 2 \cdot2 = 20$
  possible configurations for each attachment vertex.  Thus, there are
  up to $20^k$ combinations for the $k_1 + k_{12} \le k$ attachment
  vertices in $G_1$.  We get the same bounds for $G_2$.  Hence, there
  are at most $(2\rho + 1)^{2k}\cdot 400^k$ combinations for choosing
  an interface for $G_1$ and $G_2$.  For every shared attachment
  vertex, there are two shared rotations we need to set, which gives
  $25$ combinations as these rotations can take values in $[-2, 2]$.
  For the $k_{12}$ shared rotations, this gives $25^{k_{12}} \le 25^k$
  combinations, which makes $(2\rho + 1)^{2k}\cdot 10000^k$
  combinations in total.

  We start with the exponent in the factor $(2\rho + 1)^{2k}$.  The
  exponent $2k$ came from the fact that we chose rotations of $k_1 +
  k_{12} + k_2 + k_{12}$ interface paths.  Assume we have fixed the
  interface of $G_1$ except for the rotation of a single interface
  path.  As the total rotation around the face bounded by the
  interface paths is $4$ ($-4$ for the outer face) in every Kandinsky
  representation, there is no choice left for the rotation of this
  path.  Thus, we only have to choose the rotation of $k_1 + k_{12} -
  1$ interface paths in $G_1$.  The same holds for $G_2$, which gives
  $k_1 + k_{12} + k_2 + k_{12} - 2$ interface paths in total.  As
  there are $k_{12}$ shared attachment vertices, the graph $G$ has
  $k_{12} - 1$ faces that are bounded by one interface path of $G_1$
  and one interface path of $G_2$.  Assume the rotation of the
  interface paths of $G_1$ is fixed and the shared rotations are
  fixed.  Then the rotations of these $k_{12} - 1$ interface paths of
  $G_2$ are also fixed as the rotation around these faces must sum to
  $4$ ($-4$).  Thus, there are $k_{12} - 1$ additional interface paths
  whose rotation is automatically fixed.  Hence, we get the exponent
  down to $k_1 + k_2 + k_{12} - 1$ which is at most $\left\lfloor
    1.5k\right\rfloor - 1$.

  To reduce the basis of the $10000^k$ factor, first note that some
  configurations of choosing attachment rotations and $0^\circ$ flags
  are not possible.  Let $f$ be the face of $G_1$ containing all
  attachment vertices and let $v$ be an attachment vertex.  Let
  $e^\inedge$ and $e^\outedge$ be the two edges incident to $v$ and
  $f$, i.e., $\rot_f(e^\inedge, e^\outedge)$ is the attachment
  rotation at $v$.  Assume $\rot_f(e^\inedge, e^\outedge) = 2$, i.e.,
  there is an angle of $0^\circ$ at $v$.  Due to
  Property~(\ref{item:0-deg-bend-assignment}), $e^\inedge$ or
  $e^\outedge$ must have a rotation of $-1$ at the vertex $v$ in $f$
  ($\rot_f(e^\inedge[v]) = -1$ or $\rot_f(e^\outedge[v]) = -1$).  Thus
  if the attachment rotation at $v$ is $2$, at least one of the two
  $0^\circ$ flags at $v$ must be \texttt{true}.  A similar argument
  shows that an attachment rotation of $-2$ at $v$ implies that at
  least one of the $0^\circ$ flags at $v$ is \texttt{false}.  Thus,
  there are only $18$ (instead of $20$) possibilities for choosing the
  attachment rotation and the $0^\circ$ flags at an attachment vertex.
  Thus, for the exclusive attachment vertices in $G_1$ and $G_2$ we
  get $18^{k_1 + k_2}$ combinations.  Moreover, we have $18^{2k_{12}}$
  combinations for the shared attachment vertices in $G_1$ and $G_2$
  and $25^{k_{12}}$ combinations for the shared rotations.

  We show that not all these $18^{2k_{12}}\cdot 25^{k_{12}}$ need to
  be considered.  Let $v$ be a shared attachment vertex and let
  $\rot_1$ and $\rot_2$ be the attachment rotations for $v$ in $G_1$
  and $G_2$, respectively.  Let further $f$ and $f'$ be the two faces
  incident to $v$ shared by $G_1$ and $G_2$ and let $\rot_f$ and
  $\rot_{f'}$ be the corresponding shared rotations at $v$ in $f$ and
  $f'$.  Finally, let $x_f$ ($x_{f'}$) be a variable with the
  value~$1$ if the $0^\circ$ flags do not allow a $0^\circ$ angle in
  $f$ ($f'$) and the value~$0$ if they allow a~$0^\circ$ angle, which
  is the case if and only if at least one of the corresponding flags
  is \texttt{true}.  It is not hard to see, that fixing the attachment
  rotations $\rot_1$ and $\rot_2$ and the $0^\circ$ flags leaves
  $-\rot_1 - \rot_2 + 1 - x_f - x_{f'}$ possible combinations (or~$0$
  if this value is negative) to set the shared rotations when the
  result must obey the properties of a Kandinsky representation.  In
  Figure~\ref{fig:merging-step}c, $\rot_1 = -1$ and $\rot_2 = -1$
  holds.  The $0^\circ$ flags allow for a $0^\circ$ angle in $f$ but
  not in $f'$, thus $x_f = 0$ and $x_{f'} = 1$.  This leaves only two
  ways to fix the shared rotations, namely $\rot_f = 2$, $\rot_{f'} =
  0$ and $\rot_f = 1$, $\rot_{f'} = 1$.  Counting those combinations
  for each of the $18$ ways to fix the interface rotations and
  $0^\circ$ flags of $v$ in $G_1$ and $G_2$ (which can be done with a
  simple computer program) results in~$330$ combinations.  Thus, the
  $18^{2k_{12}}\cdot 25^{k_{12}}$ combinations for the shared
  attachment vertices reduce to $330^{k_{12}}$.  Hence, there are at
  most $18^{k_1 + k_2} \cdot 330^{k_{12}}$ combinations for choosing
  attachment rotations, $0^\circ$ flags, and shared rotations.  Note
  that $18^2 = 324 \le 330$ and thus we get the following.
  \[18^{k_1 + k_2}\cdot 330^{k_{12}} \le \sqrt{330}^{k_1 + k_2} \cdot
  \sqrt{330}^{2k_{12}} = \sqrt{330}^{k_1 + k_{12} + k_2 + k_{12}} \le
  \sqrt{330}^{2k} = 330^k \]

  To conclude, we get at most $330^k$ possibilities to choose all
  attachment rotations, all $0^\circ$ flags, and all shared rotations.
  Once those are chosen, at most $(2\rho + 1)^{\left\lfloor
      1.5k\right\rfloor - 1}$ ways to choose rotations of the
  interface paths remain.  Note that it is easy to list these
  combinations efficiently (without considering unnecessary
  combinations).
\end{proof}

Let $G$ be a glueable subgraph of $H$.  The \emph{cost} of an
interface class is the minimum cost of the Kandinsky representations
it contains (recall that an interface class is a set of Kandinsky
representations that have the same interface).  The \emph{cost table}
of $G$ is a table containing the cost of each interface class of~$G$.

\begin{lemma}
  \label{lem:merging-step-compute-cost-table}
  Let $G = G_1 \sqcup G_2$ be a merging step of width $k$ with maximum
  rotation $\rho$.  Given the cost tables of $G_1$ and $G_2$, the cost
  table of $G$ can be computed on $O(k\cdot (2\rho + 1)^{\left\lfloor
      1.5k\right\rfloor - 1}\cdot 330^k)$ time.
\end{lemma}
\begin{proof}
  Start with a cost table for $G$ with cost $\infty$ for every
  interface class.  We iterate over all $(2\rho + 1)^{\left\lfloor
      1.5k\right\rfloor - 1}\cdot 330^k$ compatible choices for the
  shared rotations and the interfaces of $G_1$ and $G_2$
  (Lemma~\ref{lem:merging-step-compatible-choices}).  Consider a fixed
  choice and let $[\mathcal K_1]$ and $[\mathcal K_2]$ be the chosen
  interface classes of $G_1$ and $G_2$ with cost $c_1$ and $c_2$.  In
  $O(k)$ time we can compute the interface class $[\mathcal K]$ we get
  for $G$.  If $c_1 + c_2$ is less than the current cost $[\mathcal
  K]$, we set it to $c_1 + c_2$.

  Note that the cost $c_1$ and $c_2$ imply the existence of Kandinsky
  representations $\mathcal K_1' \in [\mathcal K_1]$ and $\mathcal
  K_1' \in [\mathcal K_2]$ with cost $c_1$ and $c_2$.  These two
  Kandinsky representations can be merged (using the fixed shared
  rotations) to a Kandinsky representation $\mathcal K' \in [\mathcal
  K]$.  This representation clearly has cost $c_1 + c_2$ and thus the
  cost of $[\mathcal K]$ is at most $c_1 + c_2$.

  On the other hand, assume that there exists a Kandinsky
  representation $\mathcal K$ of $G$ with cost $c$.  Let $\mathcal
  K_1$ and $\mathcal K_2$ be the restrictions of $\mathcal K$ to $G_1$
  and $G_2$, respectively.  Let further $c_1$ and $c_2$ be the costs
  of $\mathcal K_1$ and $\mathcal K_2$, respectively.  Then $c = c_1 +
  c_2$ holds.  Moreover, the costs of the equivalence classes
  $[\mathcal K_1]$ and $[\mathcal K_2]$ are $c_1' \le c_1$ and $c_2'
  \le c_2$, respectively.  As $[\mathcal K_1]$ and $[\mathcal K_2]$
  can be merged to $[\mathcal K]$, at some point we set the cost of
  $[\mathcal K]$ $c_1' + c_2' \le c_1 + c_2 = c$.  Thus, on one hand,
  the cost of each equivalence class $[\mathcal K]$ of $G$ is never
  set to something below its actual cost, and on the other hand it is
  at some point set to a value that is at most its actual cost.
  Hence, this procedure yields the cost table of $G$.
\end{proof}

\subsection{The Algorithm}
\label{sec:algorithm}

The previous three lemmas together with a dynamic program on a sphere
cut decomposition (which is a special type of branch decomposition)
yield the following theorem.

\begin{theorem}
  \label{thm:kandinsky-general-running-time}
  An optimal Kandinsky representation of a plane graph $G$ can be
  computed in $O(n^3 + n\cdot k\cdot (2\rho + 1)^{\left\lfloor
      1.5k\right\rfloor - 1}\cdot 330^k)$ time, where $k$ is the
  branch width and $\rho$ the maximum rotation of $G$.
\end{theorem}
\begin{proof}
  Let $H$ be the plane graph.  If $H$ contains a degree-1 vertex, we
  can attach a cycle of length~$4$ to it.  Computing an optimal
  Kandinsky representation of the resulting graph and removing this
  cycle from it obviously gives an optimal Kandinsky representation of
  $H$.  Thus, we can assume without loss of generality that $H$ does
  not contain degree-1 vertices.

  In planar graphs, a branch decomposition with minimum width can be
  computed in polynomial~\cite{st-crr-94} and even
  $O(n^3)$~\cite{gt-obdpgot-08} time.  Moreover, Dorn et
  al. \cite[Theorem~1]{dpbf-eeapg-10} show that one can compute a
  sphere cut decomposition of width $k$ from a given branch
  decomposition of width $k$ in $O(n^3)$ time, if $G$ does not contain
  degree-1 vertices.  Without defining sphere cut decomposition
  precisely, it is essentially a rooted binary tree $\mathcal T$
  (every node has two children or is a leaf) with a bijection between
  the edges of $H$ and the leaves of $\mathcal T$ such that the
  following property holds.  For every node $\mu$ of $\mathcal T$, the
  edges of $H$ corresponding to leaves that are ancestors of $\mu$
  induce a glueable subgraph of $H$.  Denote this subgraph by $G_\mu$.

  Clearly, this implies that for an inner node $\mu$ with children
  $\mu_1$ and $\mu_2$, we get a merging step $G_\mu = G_{\mu_1}\sqcup
  G_{\mu_2}$.  We process the inner nodes of $\mathcal T$ bottom up to
  compute the cost table of $G_\mu$ for every node $\mu$.  If a child
  $\mu_i$ (for $i = 1, 2$) of $\mu$ is a leaf, it corresponds to a
  single edge for which the cost tables are trivially known.
  Otherwise, we already processed the child $\mu_i$ and thus know the
  cost table of $G_{\mu_i}$.  Hence, by
  Lemma~\ref{lem:merging-step-compute-cost-table}, we can compute the
  cost table of $G_\mu$ in $O(k\cdot (2\rho + 1)^{\left\lfloor
      1.5k\right\rfloor - 1}\cdot 330^k)$ time.  Doing this for every
  inner node of $\mathcal T$ gives the claimed running time, since
  $\mathcal T$ contains $O(n)$ inner nodes.  Moreover, for the root
  $\tau$, we have $G_\tau = H$.  Thus, after processing the root
  $\tau$, we know the cost of an optimal Kandinsky representation of
  $H$.  To actually compute an optimal Kandinsky representation of $H$
  (and not only its cost) one simply has to track the interface
  classes that lead to the optimal solution through the dynamic
  program.
\end{proof}

We get the following bounds for the maximum rotation $\rho$ of a
graph.

\begin{lemma}
  \label{lem:rotation-interface-paths}
  Let $G$ be a graph with Kandinsky representation $\mathcal K$.  Let
  $\Delta_F$ be the maximum face degree of $G$ and let $\rho$ be the
  maximum absolute rotation of interface paths of glueable subgraphs
  of $G$.  The following holds.
  \begin{compactitem}
  \item $\rho \le m + \Delta_F - 2$, if $\mathcal K$ is an optimal
    Kandinsky representation.
  \item $\rho \le (b+1)\cdot \Delta_F - b - 2$, if $\mathcal K$ is a
    $b$-bend Kandinsky representation.
  \end{compactitem}
\end{lemma}
\begin{proof}
  First note that every interface path of a glueable subgraph is a
  subpath of a face of $G$.  Thus, interface paths have length at most
  $\Delta_F - 1$.  We show that the maximum rotation of a path of this
  length satisfies the claimed bounds.  Proving that the absolute
  value of the minimum rotation also satisfies these bounds is
  symmetric.  

  Consider the case that $\mathcal K$ is an optimal Kandinsky
  representation.  As $G$ admits a 1-bend
  representation~\cite{fkk-2dpg-97}, there exists a representation
  with $m$ bends and an optimal representation $\mathcal K$ has at
  most $m$ bends.  Thus, the edges on the interface path have at most
  $m$ bends contributing rotation at most $m$.  An interface path of
  length (in terms of number of edges) at most $\Delta_F - 1$ has at
  most $\Delta_F - 2$ inner vertices.  If the rotation of each inner
  vertex is at most $1$ we get the claimed inequality $\rho \le m +
  \Delta_F - 2$.  Consider a vertex $v$ with rotation $2$.  Due to
  Property~(\ref{item:0-deg-bend-assignment}), at least one of the two
  edges in the path incident to $v$ must have rotation $-1$.  Thus, we
  can account rotation~1 even for vertices with rotation~$2$, yielding
  $\rho \le m + \Delta_F - 2$.

  In case $\mathcal K$ is a $b$-bend Kandinsky representation, the
  rotation contributed by the edges is at most $b\cdot (\Delta_F -
  1)$.  Together with the $\Delta_F - 2$ upper bound for the vertices,
  this gives $\rho \le b\cdot \Delta_F - b + \Delta_F - 2 = (b +
  1)\cdot \Delta_F - b - 2$.
\end{proof}

We get the following corollaries by plugging the bounds of
Lemma~\ref{lem:rotation-interface-paths} into
Theorem~\ref{thm:kandinsky-general-running-time}, using that the
branch width of series-parallel graphs is~$2$, and that the branch
width of planar graphs is in $O(\sqrt{n})$ (in fact, the branch width
of a planar graph is at most $2.122\sqrt{n}$~\cite{ft-nubdpg-06}).

\begin{corollary}
  Let $G$ be a plane graph with maximum face-degree $\Delta_F$, and
  branch width $k$.  An optimal Kandinsky representation can be
  computed in $O(n^3 + n\cdot k\cdot (2m + 2\Delta_F -
  3)^{\left\lfloor 1.5k\right\rfloor - 1}\cdot 330^k)$ time.  An
  optimal $b$-bend Kandinsky representation can be computed in $O(n^3
  + n\cdot k\cdot ((2b + 2)\cdot \Delta_F - 2b - 3)^{\left\lfloor
      1.5k\right\rfloor - 1}\cdot 330^k)$ time.
\end{corollary}
\begin{corollary}
  For series-parallel graphs an optimal Kandinsky representation can
  be computed in $O(n^3)$ time.
\end{corollary}
\begin{corollary}
  For plane graphs an optimal Kandinsky representation can be computed
  in $2^{O(\sqrt{n} \log n)}$ time.
\end{corollary}

\section{Conclusion}
\label{sec:conclusion}

In this paper we have shown that bend minimization in the Kandinsky
model is NP-complete, thus answering a question that was open for
almost two decades.  The proof also extends to the case that every
edge may have at most one bend and for the case that empty faces are
allowed.

On the positive side, we gave an algorithm with running time $2^{O(k
  \log n)}$ for graphs of bounded branch width.  In fact, the problem
is FPT with respect to $k + b + \Delta_F$, where $k$ is the branch
width, $b$ is the maximum number of bends on a single edge in the
drawing and $\Delta_F$ is the size of the largest face in the
combinatorial embedding.  For general planar graphs this gives a
subexponential exact algorithm with running time $2^{O(\sqrt{n}\log
  n)}$.

We leave open the question whether the number of parameters used to
obtain an FPT algorithm can be decreased.  Is the problem $W[1]$-hard
when parameterized by branch width only?

\bigskip

\noindent\textbf{Acknowledgments.}  We thank Therese Biedl for
discussions.

\bibliographystyle{abbrv}
\bibliography{kandinsky}

\end{document}